\documentclass[12pt]{article}
\usepackage{natbib,geometry,graphicx,epsfig,pdfpages,longtable,url,amsmath,bbm,amssymb,amsthm, afterpage, stmaryrd}
\usepackage{algorithm}
\usepackage{algpseudocode}

\newtheoremstyle{dotless}{0cm}{}{\itshape}{}{\bfseries}{}{ }{}
\theoremstyle{dotless}
\newtheorem{theorem}{Theorem}
\newtheorem{proposition}{Proposition}
\newtheorem{lemma}{Lemma}
\theoremstyle{remark}
\newtheorem{remark}{Remark}

\theoremstyle{dotless}
\newtheorem{assumption}{Assumption}

\begin{document}
\title{Easily Computed Marginal Likelihoods from Posterior Simulation Using the THAMES Estimator}
\author{Martin Metodiev \\ Friedrich-Alexander-Universit\"{a}t Erlangen-N\"{u}rnberg \\ Universit\'{e} Paris Cit\'{e} \and
Marie Perrot-Dock\`{e}s \\ Universit\'{e} Paris Cit\'{e} \and
Sarah Ouadah \\ AgroParisTech \and
Nicholas J. Irons \\ University of Washington \and 
Adrian E. Raftery\thanks{Corresponding author.} \\ University of Washington }
\date{\today}
\maketitle 

\begin{abstract}
We propose an easily computed estimator of marginal likelihoods from
posterior simulation output, via reciprocal importance sampling,
combining earlier proposals of DiCiccio et al (1997) and Robert and
Wraith (2009).
This involves only the unnormalized posterior densities from the 
sampled parameter values, and does not involve additional simulations
beyond the main posterior simulation, or additional complicated calculations.
It is unbiased for the reciprocal of the marginal likelihood,
consistent, has finite variance, and is asymptotically normal. 
It involves one user-specified control parameter, and we derive an 
optimal way of specifying this. We illustrate it with several numerical
examples.
\end{abstract}

\baselineskip=18pt

\section{Introduction}
A key quantity in Bayesian model selection is the marginal likelihood,
also known as the evidence, the normalizing constant of the posterior density,
or the integrated likelihood. Consider a statistical model 
with parameter vector $\theta$ and data $\mathcal{D}$. Let $L(\theta) = p(\mathcal{D}|\theta)$
be the usual likelihood, and $\pi(\theta)$ be the prior distribution of 
$\theta$. Then $Z = p(\mathcal{D}) = \int L(\theta) \pi(\theta) d\theta$ is the
marginal likelihood. 

The marginal likelihood plays a key role in defining Bayes factors.
Consider two models $M_1$ and $M_2$ with marginal likelihoods $Z_1$ and
$Z_2$. Then the  Bayes factor (or ratio of posterior to prior odds) for 
model $M_1$ against $M_2$ is $B_{1,2} = Z_1 / Z_2$.

The marginal likelihood is also a critical quantity for Bayesian model 
averaging (BMA). Consider $K$ models, $M_1,\ldots,M_K$, with
prior model probabilities $\Pi_k$ (which add up to 1), and marginal
likelihoods $Z_k$.
Suppose $Q$ is a quantity of interest, such as a
parameter or a future observation to be predicted.
Then the BMA posterior distribution of $Q$ is
\begin{equation}
p(Q|\mathcal{D}) = \sum_{k=1}^K p(Q|\mathcal{D},M_k) p(M_k|\mathcal{D}),
\end{equation}
where $p(M_k|\mathcal{D})$ is the posterior model probability of $M_k$, which satisfies
$p(M_k|\mathcal{D}) \propto \Pi_k Z_k$ and $\sum_{k=1}^K p(M_k|\mathcal{D}) =1 $.
So $p(Q|\mathcal{D}) = \sum_{k=1}^K p(Q|\mathcal{D},M_k) \Pi_k Z_k / \sum_{k=1}^K \Pi_k Z_k$.

Finally, the most likely model {\it a posteriori} is the one that 
maximizes $\Pi_k Z_k$.
Choosing it minimizes the model selection error rate on average over the prior 
\citep{Jeffreys1961}.
Often the prior over the model space is chosen to be uniform,
in which case $\Pi_k = 1/K, \; \forall k$.
In this case, Bayesian model selection by choosing the most likely model
a posteriori boils down to choosing the model with the largest $Z_k$,
and hence involves only the marginal likelihoods.

Bayesian models are often estimated using Monte Carlo methods in which
a sample of values of $\theta$ is simulated from the posterior distribution.
The most common class of such methods is Markov chain Monte Carlo (MCMC).
Perhaps surprisingly, estimating the marginal likelihood from the output
of MCMC and other posterior simulation methods has turned out not to be
straightforward. Many different methods have been proposed, and none of them
is widely considered to be generally the best. \cite{Llorente&2023} provide
a comprehensive review of such methods, describing 16 different methods and,
remarkably, cite over 20 {\it other} review articles!

We seek a method that is accurate, generic and simple for estimating the
marginal likelihood from posterior simulation output.
We take this to mean that it gives accurate estimates of the marginal 
likelihood, uses posterior simulation output for just the one model being
analyzed, uses only likelihoods and prior densities of the sampled
values of $\theta$, and does not need additional simulations or 
complicated calculations. 

Some well-known methods do not satisfy our desiderata. These include
Chib's method \citep{Chib1995}, which requires complicated additional
calculations, bridge sampling \citep{MengWong1996}, which requires simulations
from two models, importance sampling, which requires additional 
simulations, and nested sampling \citep{Skilling2006}, which involves
other simulations.
They also include the harmonic mean of the likelihoods \citep{NewtonRaftery1994}, which is unbiased and consistent, but has infinite variance and is unstable, as pointed out by the original authors.

Arguably, the only methods that are accurate, generic and simple for estimating 
the marginal likelihood from MCMC by our definition  are versions of 
reciprocal importance sampling (RIS) \citep{GelfandDey1994}. 
These are based on the identity:
\begin{equation}
Z^{-1} = E_{\theta}  \left[ \frac{h(\theta)}{L(\theta) \pi(\theta)}
 \bigg| \mathcal{D} \right], 
\label{eq-risidentity}
\end{equation}
where $h(\theta)$ is a (normalized) probability density function (pdf) 
over the posterior support.
Remarkably, this holds for any pdf $h(\theta)$.
This leads to the estimator
\begin{equation}
\hat{Z}^{-1} = \frac{1}{T} \sum_{t=1}^{T}
\frac{h(\theta^{(t)})}{L(\theta^{(t)}) \pi(\theta^{(t)})} , 
\label{eq-risest}
\end{equation}
where $\theta^{(1)},\dots,\theta^{(T)}$ are simulated from the posterior using MCMC or another method. This estimator has good properties in general, provided that the tails of
the distribution $h(\theta)$ are thin enough in all directions. 
It can be hard to choose $h(\theta)$ so that it both overlaps substantially with
the posterior distribution (needed for efficiency) and has thin enough
tails, especially in higher dimensions. We propose a choice of 
$h(\theta)$ that leads to easily computed estimates and is optimal or
near optimal in a certain sense.

The paper is organized as follows. 
In Section \ref{sect-ris} we discuss reciprocal importance sampling and
its properties.
In Section \ref{sect-THAMES} we describe our proposed choice of 
$h(\theta)$ and derive some of its properties.
In Section \ref{sect-examples} we give several numerical examples, including
a multivariate Gaussian example, a Bayesian regression example, a non-Gaussian case, and a Bayesian hierarchical model.
We conclude in Section \ref{sect-discussion} with a discussion.

\section{Reciprocal Importance Sampling}
\label{sect-ris}

In general, the RIS estimator of the marginal likelihood is defined by
Equation \eqref{eq-risest}. This has several good properties.
It is unbiased, in the sense that $E[\hat{Z}^{-1}] = {Z}^{-1}$,
where the expectation is over the posterior distribution of $\theta$.
It is also strongly simulation-consistent, in the sense that $\hat{Z}^{-1} \longrightarrow {Z}^{-1}$ almost surely as $T \longrightarrow \infty$.

In addition, the RIS estimator of the reciprocal marginal likelihood,
$\hat{Z}^{-1}$, has finite variance and is asymptotically normally 
distributed as $T \longrightarrow \infty$ {\it if} the tails of 
$h(\theta)$ are thin enough. Specifically, this requires that
\begin{equation}
\int \frac{h(\theta)^2}{L(\theta)\pi(\theta)} d\theta < \infty.
\end{equation}

It is hard to choose $h(\theta)$ so that it both overlaps substantially with the area of the parameter space with high posterior density, which is needed
for efficiency, and so that it also has thin enough tails, which is needed for 
finite variance. The difficulty grows as the dimension increases.

Two choices of $h(\theta)$ in the literature deserve attention.
\citet{DiCiccio&1997} proposed $h(\theta) = MVN (\theta; \hat{\theta}, \hat{\Sigma})$,
where $\hat{\theta}$ is the posterior mean or mode, and
$\hat{\Sigma}$ is an estimate of the posterior covariance matrix.
This overlaps nicely with $L(\theta) \pi(\theta)$, but its tails
may not be thin enough when the posterior is asymmetric or the parameter
is high-dimensional.

To remedy the problem of the tails possibly being too thick, 
\citet{DiCiccio&1997} proposed truncating it, using instead
$h(\theta) = TMVN_{A}(\hat{\theta}, \hat{\Sigma})$, a multivariate normal
distribution truncated to the set $A$, where
\begin{equation}
A = \{\theta: (\theta - \hat{\theta})^T \hat{\Sigma}^{-1}
 (\theta - \hat{\theta}) < c^2 \}.
\label{eq-A}
\end{equation}
Thus $A$ is an ellipsoid with radius $c$ and volume
\begin{equation}
V(A) = c^d \pi^{d/2} |\hat{\Sigma}|^{1/2} / \Gamma \left( \frac{d}{2} + 1 \right) .
\label{eq-V}
\end{equation}
Truncating the distribution ensures that the estimator $\hat{Z}^{-1}$
has finite variance.
They found that the truncation improved the performance of the RIS estimator.
However, with high-dimensional parameters, the result might be 
sensitive to the specification of $\Sigma$.

\citet{RobertWraith2009} proposed setting $h(\theta)$ to be a uniform 
distribution on the convex hull of simulated MCMC parameters values in
the $\alpha$-HPD region, namely the highest posterior density region 
containing a proportion $\alpha$ of the sampled parameter values.
They considered the values $\alpha=0.1$ and 0.25. They applied it to 
a two-dimensional toy example where it performed well.

However, as far as we know, 
the method has not yet been fully developed for realistic, 
higher-dimensional situations. For example, we know of no simple 
way to compute the volume of the convex hull of a set of points in
higher dimensions, which is required for the method in general. 
It is also not clear how best to choose $\alpha$
nor how sensitive the method would be to $\alpha$ in higher dimensions.
It has been used in a higher-dimensions application by 
\citet{Durmus&2018}, but this involved comparing competing models defined
on the same parameter space, thus avoiding the need to calculate the 
volume of $A$, which canceled out in Bayesian model comparisons. Calculating the volume
of $A$ may be the most difficult part of this method in general.

\section{Estimating the marginal likelihood}
\label{sect-THAMES}
\subsection{Estimating the marginal likelihood with THAMES}
We propose combining the proposals of \citet{DiCiccio&1997} and
\citet{RobertWraith2009} to obtain a method that we believe satisfies all our 
desiderata. We propose specifying $h(\theta)$ to be a uniform distribution,
but to be uniform over the set $A$ defined in Equation \eqref{eq-A}, rather than over
a convex hull of points. This resolves the problem of computing the volume
of $A$, since this is given analytically by Equation \eqref{eq-V}. If $A$ is not a subset of the posterior support, for example if the posterior support is constrained, we adjust the volume of $A$ by a simple Monte Carlo approximation.

This yields the
estimator 
\begin{equation}
 \hat{Z}^{-1} = \frac{1}{V(A)T}
 \sum_{\substack{t=1 \\ \theta^{(t)} \in A}}^T 
  \frac{1}{L(\theta^{(t)}) \pi(\theta^{(t)})}.
\label{eq-Zhat}
\end{equation}
Thus $\hat{Z}$ is a truncated harmonic mean of the unnormalized
posterior densities, $L(\theta^{(t)}) \pi(\theta^{(t)})$.\footnote{Recall
that the unstable harmonic mean estimator described by \citep{NewtonRaftery1994} was quite different, not being truncated, and being a harmonic mean of the likelihoods rather than the unnormalized posterior density values.}
We call it the Truncated HArmonic Mean EStimator, or THAMES.

The THAMES, $\hat{Z}^{-1}$, has several desirable properties. It is simple to compute,
involving only the prior and likelihood values of the sampled parameter
values. In fact it involves only the product of the prior and likelihood values,
namely the unnormalized posterior densities of the sampled parameter values.
It is unbiased as an estimator of $Z^{-1}$. It is also 
simulation-consistent, in the sense that 
$\hat{Z}^{-1} \longrightarrow {Z}^{-1}$ almost surely
as $T \longrightarrow \infty$, by the strong law of large numbers.
Its variance (over simulation from the posterior given the data $\mathcal{D}$)
is finite provided that 
\begin{equation}
\int_A \left( L(\theta) \pi(\theta) \right)^{-1} d\theta < \infty,
\label{eq-finitevariance}
\end{equation}
which will usually hold since $A$ is a bounded set in $\mathbbm{R}^d$. In fact, it suffices that the likelihood and the prior are continuous with respect to $\theta$ and strictly positive on the closure of $A$.
If Equation \eqref{eq-finitevariance} holds, $\hat{Z}^{-1}$ is asymptotically normal
(again as the number of parameter values simulated increases),
by the Lindeberg central limit theorem. Note that asymptotic normality
holds on the scale of $\hat{Z}^{-1}$, and not exactly on other scales such as
$\hat{Z}$ or $\log (\hat{Z})$. 

If the posterior simulation method yields independent draws, then 
Var($\hat{Z}^{-1}$) can be estimated directly as the empirical variance
of the values of 
$\left( L(\theta^{(t)}) \pi (\theta^{(t)}) \mathbbm{1}(\theta^{(t)} \in A) \right)^{-1}$, divided by $V(A)^2$. 
If MCMC is used, successive simulations from the posterior will in general
not be independent. A central limit theorem will still hold, but the variance
needs to take account of the serial dependence. This can be done approximately
by computing the variance based on serial independence and multiplying it
by an estimate of the spectral density of the sequence at zero.
For example, if the sequence of values of $1/ \left( L(\theta) \pi(\theta) \right)$ can be approximated by a first-order autoregressive model with parameter
$\phi$, then this would be approximately $1/(1-\phi)^2$.
An alternative would be to thin the sequence enough that the resulting
subsequence is approximately uncorrelated and then use the variance 
based on assuming independence. A different approach was taken by
\citet{Fruhwirth2004}.

Note that an approximate normal confidence interval can be obtained for $\hat{Z}^{-1}$,
because that is the scale on which a central limit theorem holds.
This could be turned into a confidence interval for $\hat{Z}$ by taking the
reciprocals of the ends of the normal confidence intervals for $\hat{Z}^{-1}$;
the resulting confidence interval would not be symmetric. 
The same could be done for $\log(\hat{Z})$ in a similar manner. 

\subsection{Optimal choice of control parameter, $c$}
\label{ssec:copt}
We now address the question of how to choose the radius $c$ of the ellipse
that specifies the THAMES in Equation \eqref{eq-A}. Ignoring serial correlation between
simulated values of the parameters, we suggest choosing $c$ to 
minimize the estimated variance of $\hat{Z}^{-1}$. This could be done
empirically by computing $\hat{Z}^{-1}$ for a range of values of $c$,
estimating Var($\hat{Z}^{-1}$) for each value of $c$, and optimizing it
over $c$ by a grid search or a one-dimensional numerical optimization method. 

It is possible to obtain analytic results in the case where the posterior
distribution is normal. This is of considerable interest as the posterior
distribution is asymptotically normal in many common situations,
including some where standard regularity conditions do not hold
\citep{HeydeJohnstone1979,Ghosal2000,Shen2002,Miller2021}.
In this case the THAMES has finite variance since the posterior density, and thus the product of the likelihood and the prior, is continuous with respect to $\theta$ and strictly positive everywhere.

We want to minimize the variance of the THAMES. Due to our assumption of independence of all of the successive MCMC simulations, this variance can be simplified to \begin{equation}
    Var(\hat{Z}^{-1}|\mathcal{D})=\frac{1}{T}\cdot \frac{1}{Z^2}\cdot SCV(d,c).
\end{equation} 

Here $SCV(d,c)$ denotes \begin{equation}
    SCV(d,c):=\frac{Var_{\theta^{(1)}}\left(\left.\frac{\mathbbm{1}_A(\theta^{(1)})/V(A)}{L(\theta^{(1)})\pi(\theta^{(1)})}\right|\mathcal{D}\right)}{E_{\theta^{(1)}}\left(\left.\frac{\mathbbm{1}_A(\theta^{(1)})/V(A)}{L(\theta^{(1)})\pi(\theta^{(1)})}\right|\mathcal{D}\right)^2},
\end{equation}  the squared coefficient of variation of the first term of the THAMES. Since the variance is a product of $\frac{1}{T},\frac{1}{Z^2}$ and $SCV(d,c)$, minimizing $SCV(d,c)$ with respect to $c$ is equivalent to minimizing the variance of the THAMES. 

We derive a statement about the optimal choice of $c$ by assuming that the posterior covariance matrix $\Sigma$ and the posterior mean $m$ can be provided by a stochastic oracle. The THAMES can then be defined using \begin{equation}
A_{or} := \{\theta: (\theta - m)^T \Sigma^{-1}
 (\theta - m) < c^2 \}.
\label{eq-A_or}
\end{equation} Interestingly, in this case the $SCV$ depends neither on the data, $\mathcal{D}$, nor on the number of samples from the posterior, $T$. Of course, this is rarely the case in practice. However, plugging in consistent estimators of $m,\Sigma$ approximately gives the same results if the number of samples from the posterior is large enough. This is due to the continuous mapping theorem.

The proofs of these results are given in Appendix 1.

\; \;
\begin{assumption}For the following theorems it is assumed that we can ignore serial correlation (i.e. we assume independence of all of the successive MCMC simulations) and that the posterior distribution is normal with mean $m\in\mathbbm{R}^d$ and a positive definite covariance matrix \\$\Sigma\in\mathcal{M}_{d\times d}(\mathbbm{R})$. We further assume that the THAMES is defined on $A_{or}$.\end{assumption}

\begin{theorem}\label{thm-limit_behav_c_d}
There exists a unique radius $c_d\in(0,\infty)$ such that the ellipse $A_{or}$ with radius $c_d$ minimizes the variance of the THAMES. This value $c_d$ does not depend on the posterior mean or covariance matrix. It satisfies $c_d=\sqrt{d+L_d}$, where the optimal shifting parameter $L_d\geq0$ is a sequence for which $\frac{L_d}{d}\stackrel{d\to\infty}\to0$ holds.\end{theorem}

\begin{remark}\label{rem-thm1} Theorem \ref{thm-limit_behav_c_d} ensures that the optimal radius $c_d$ is asymptotically equivalent to $\sqrt{d}$.  In fact, our calculations suggest that $c_d=\sqrt{d+L_d}$ can be approximated by $\sqrt{d+1}$: \\The shifting parameter $L_d$ approaches 1 exponentially fast (Figure \ref{fig-l_apprx_sol}), with \begin{alignat*}{2}
    |1-L_d|&\leq 0.5\quad &d\geq1,\\|1-L_d|&\leq 0.05\quad &d\geq10,\\|1-L_d|&\leq 0.005 \quad &d\geq100.
\end{alignat*} For this reason we recommend choosing the radius $c=\sqrt{d+1}$.
\end{remark}

\begin{theorem}\label{thm-limit_behav_SCV}

The following statements hold for the $SCV$:

\begin{enumerate}
    \item\label{it-thm1-L_robust} For any choice of the shifting parameter $L,L\in \mathbbm{R}$ there exists an $\varepsilon>0$ such that\begin{equation}
    1-\varepsilon\leq \frac{SCV(d,\sqrt{d+L_d})}{\sqrt{(d+2)\pi/4}}\leq \frac{SCV(d,\sqrt{d+L})}{\sqrt{(d+2)\pi/4}}\leq 2+\varepsilon,
\end{equation} for all but finitely many $d$. Thus choosing the radius $\sqrt{d+L}$ results in an $SCV$ that is both asymptotically at most twice as large as the optimal $SCV$ and is of order $\sqrt{d}$.

\item\label{it-thm1-ineq_sqrt_d_plus_1} The following inequality for the $SCV$ can be given for choosing the radius $c=\sqrt{d+1}$:\begin{align}0.63\sqrt{(d+2)\pi/4}-1&\leq SCV(d,\sqrt{d+L_d})\\&\leq 
    SCV(d,\sqrt{d+1})\leq 1.09\cdot 2\sqrt{(d+2)\pi/4}-1
\end{align} This inequality holds for all $d\geq1$.\end{enumerate}\end{theorem}

\begin{remark}\label{rem-stm2-thm1} 
    Statement \ref{it-thm1-L_robust} of Theorem \ref{thm-limit_behav_SCV} shows that $SCV(d,c)$ is increasing with order $\sqrt{d}$ as $d\to\infty$, both in our choice $c=\sqrt{d+1}$ and the optimal choice $c_d$. This is strongly  supported numerically by Figure \ref{fig-thames_sqrt_plot_thm1}. This also shows that $SCV(d,\sqrt{d+1})$ is extremely close to the theoretical lower bound of $SCV$.  Further, any choice of the shifting parameter $L$ used to define the radius $\sqrt{d+L}$  is asymptotically at most twice as bad as any optimal solution in terms of the $SCV$. This suggests some robustness of our estimator with respect to the choice of $L$. 
    
    Also note that $\sqrt{T}(\hat{Z}^{-1}-Z^{-1})$ is asymptotically normal. However, it is usually the case that we want to estimate the logarithm of the marginal likelihood. We can do this by using the estimator $-\log(\hat{Z}^{-1})$. The asymptotic behavior of this estimator can be determined by a delta-method approximation:\begin{align}
        \sqrt{T}(-\log(\hat{Z}^{-1})-\log(Z))\sim \mathcal{N}(0,SCV(d,c)).
    \end{align}
 Thus the $SCV$ roughly reflects the variance of our estimator on the log likelihood scale.
 \end{remark}
\; \;

\begin{remark}\label{rem-stm3-thm1}    
    Statement \ref{it-thm1-ineq_sqrt_d_plus_1} of Theorem \ref{thm-limit_behav_SCV} gives a very rough theoretical guarantee: For any dimension $d\geq1$, the $SCV$ obtained by choosing our recommendation for the radius, $\sqrt{d+1}$, and the $SCV$ obtained by choosing the optimal radius, $c_d$, can be bounded by an affine transform of $\sqrt{d+2}$. However, our calculations suggest that the $SCV$ in the point $c=\sqrt{d+1}$ has an asymptotically optimal performance. 
    
    This is well illustrated numerically by Figure \ref{fig-thames_scv_and_l}, in which
    $SCV(d,\sqrt{d+L})$ is plotted against $d$ for several values of $L$. The lower bound from Remark \ref{rem-stm2-thm1} is added. The $SCV$ of our estimator (purple) seems to approach the lower bound, while the $SCV$ resulting from other choices of $L$ is larger, but not more than twice as large as this bound,  even for small $d$.
\end{remark} 

So far, we have given results for the idealized situation where the posterior distribution is exactly normal. We now give a result for the much more common and realistic situation where the posterior distribution is only asymptotically normal.

\begin{theorem}\label{thm-n_limit_approx} Let $p_n(\theta|\mathcal{D}_n)$ be a sequence of posterior densities with data $\mathcal{D}_n$, posterior covariance matrix $\Sigma_n$, posterior mean $m_n$ and as $SCV$ denoted by $SCV_n$. Then, if \begin{align}\label{eq-unif_conv_pdf}
    |\Sigma_n|^{\frac{1}{2}}p_n\left(\left.\Sigma_n^{\frac{1}{2}}\cdot\theta+m_n\right|\mathcal{D}_n\right)\stackrel{n\to\infty}\to |\Sigma|^{\frac{1}{2}}p\left(\left.\Sigma^{\frac{1}{2}}\cdot\theta+m\right|\mathcal{D}\right)
\end{align} uniformly in $\theta$ on all compact subsets of $\mathbbm{R}^d$, it is the case that \begin{align}\label{eq-conv_scv}
    SCV_n(d,c) \stackrel{n\to\infty}\to SCV(d,c)
\end{align} uniformly in $c$ on all compact subsets of $(0,\infty)$. In particular, for any $b\geq c_d\geq a>0$, \begin{align}\label{eq-conv_argmin_scv}
    (c_d)_n \in \textup{argmin}_{c\in[a,b]}SCV_n(d,c)\;\forall n\Rightarrow \lim_{n\to\infty} (c_d)_n=c_d.
\end{align}\end{theorem} 

\begin{remark}

    We have already stated that the normal case is important because the posterior distribution is often asymptotically normal when the size of the data, $n$, is large. Theorem \ref{thm-n_limit_approx} assures that our results still hold in this limiting case, under some assumptions: 

    If the convergence of the normalized posterior pdf is uniform in $\theta$ (Equation \eqref{eq-unif_conv_pdf}), our statements about the limiting behaviour of the $SCV$ (Theorem \ref{thm-limit_behav_SCV} and Remarks \ref{rem-stm2-thm1}-\ref{rem-stm3-thm1}) still hold approximately when $n$ is large (Equation \eqref{eq-conv_scv}). If additionally any optimal radius $(c_d)_n$ does not converge to zero or infinity, any result about $c_d$ (Theorem  \ref{thm-limit_behav_c_d} and Remark \ref{rem-thm1}) also holds approximately when $n$ is large (Equation \eqref{eq-conv_argmin_scv}).
    
    Let $H_0$ denote the Fisher information matrix. Reformulating Equation \eqref{eq-unif_conv_pdf} by replacing $\Sigma_n$ by $\frac{1}{n}H_0^{-1}$, to which it is asymptotically equivalent, $\frac{1}{n}H_0^{-1}$, gives a statement that has been proven under a variety of assumptions (e.g., \cite[Theorem 4]{Miller2021}), except that in these results the type of convergence is usually not uniform convergence, but a weaker type of convergence, such as convergence in distribution or convergence in total variation. 
    
    Additional assumptions can be placed on the pdfs of the sequence of distributions such that convergence in distribution implies uniform convergence of the pdfs. For example, if the pdfs are asymptotically equicontinuous and we have convergence in distribution, the convergence of the pdfs is uniform \cite[Theorem 1]{Sweeting1986}. Note that in this case there is no problem if the parameter space is constrained: Uniform convergence of the pdfs implies that $A_n$ is a subset of the posterior support if $n$ is large enough.
\end{remark}

\begin{remark}\label{rem-heuristic-normality}    
    Due to the assumption of normality it is the case that when choosing the optimal radius $c_d=\sqrt{d+L_d}$, the probability of a term of the THAMES in $\theta^{(t)}$ not being set to 0 is equal to \begin{align}\label{eq-heuristic_calculation}
    \mathbb{P}(\theta^{(t)}\in A_{or})=\mathbbm{P}((\theta^{(t)} - m)^T \Sigma^{-1}
 (\theta^{(t)} - m) < d+L_d)=\chi^2(d+L_d;d),
\end{align} the CDF of the $\chi^2$-distribution with $d$ degrees of freedom evaluated at $d+L_d$. It approaches $50\%$ due to Theorem \ref{thm-limit_behav_c_d} (compare Figure \ref{fig-thames_hpd_region}). Thus the algorithm sets about $50\%$ of the highest terms in Equation \eqref{eq-Zhat} to 0. This means that for a large number of samples $T$ and given the normality assumption, our algorithm is similar to the following method:

Instead of checking whether $\theta^{(t)}\in A_{or}$ directly, one can set roughly $50\%$ of the highest terms of the THAMES, the terms not included in the Heighest Posterior Density (HPD) region of size $50\%$, to 0. 

We can check how this method performs in the normal case by setting \begin{align}\label{eq-heuristic_radius}
    c=\sqrt{(\chi^2)^{-1}(50\%,d)}\simeq d\left(1-\frac{2}{9d}\right)^2,
\end{align} the $50\%$-quantile of the $\chi^2$-distribution with $d$ degrees of freedom. This corresponds to setting roughly $50\%$ of the highest terms of the THAMES to 0 and dividing by the volume of the ellipse $A_{or}$ with radius defined in Equation \eqref{eq-heuristic_radius}. Figure \ref{fig-thames_ratio_opt_vs_heuristic} shows the ratio between the $SCV$ when taking this approach and the SCV when using the optimal radius. The ratio is decreasing to 1, so the heuristic performs quite well, even for small $d$.

Note that the calculations in Equation \eqref{eq-heuristic_calculation} still hold asymptotically for large $d$ even if $\theta^{(t)}|\mathcal{D}$ is not normal, but its elements $\theta^{(t)}_1,\dots,\theta^{(t)}_d$ satisfy a central limit theorem, e.g., if the entries of $\sqrt{\Sigma}(\theta^{(t)} - m)$ are independent given the data. This could justify using the heuristic in the case when asymptotic posterior normality does not hold, though the estimator is also necessarily biased in this case, with the bias vanishing as $d$ increases.
\end{remark}

\begin{remark}\label{rem-sigma_mustbe_pos_def}
    It is assumed that the covariance matrix of the posterior distribution is positive definite. This assumption is necessary since otherwise a posterior density with respect to the Lebesgue-measure on $\mathbbm{R}^d$ would not exist. On the other hand this assumption is not restrictive, since the same estimation procedure can be applied to the lower dimensional subspace of $\mathbbm{R}^d$ on which a density is defined.
\end{remark}

We can illustrate the relationship between the THAMES and the harmonic mean estimator defined by \citet{NewtonRaftery1994} using the toy example from Figure \ref{fig-thames_conv_comp}. It was calculated using the same model as the one introduced in Section \ref{ssec:multivariateGaussian} with the dimensions of the parameter space $d=2$, but by setting the data set to $\mathcal{D}\equiv0$ to ensure stability of the estimator on the inverse likelihood scale. 

The pdf of the Uniform distribution on the ellipse is essentially used as a rejection rule: Values with a very low posterior (and therefore high inverse posterior) are rejected, while high-density values are accepted. A balance between the volume of the ellipse and the percentage of the rejected posterior sample needs to be found to ensure optimal performance. The harmonic mean estimator does not have this rejection rule, so sample points with a high posterior can lead to massive jumps.

\begin{figure}
\centering
\includegraphics[scale=.5]{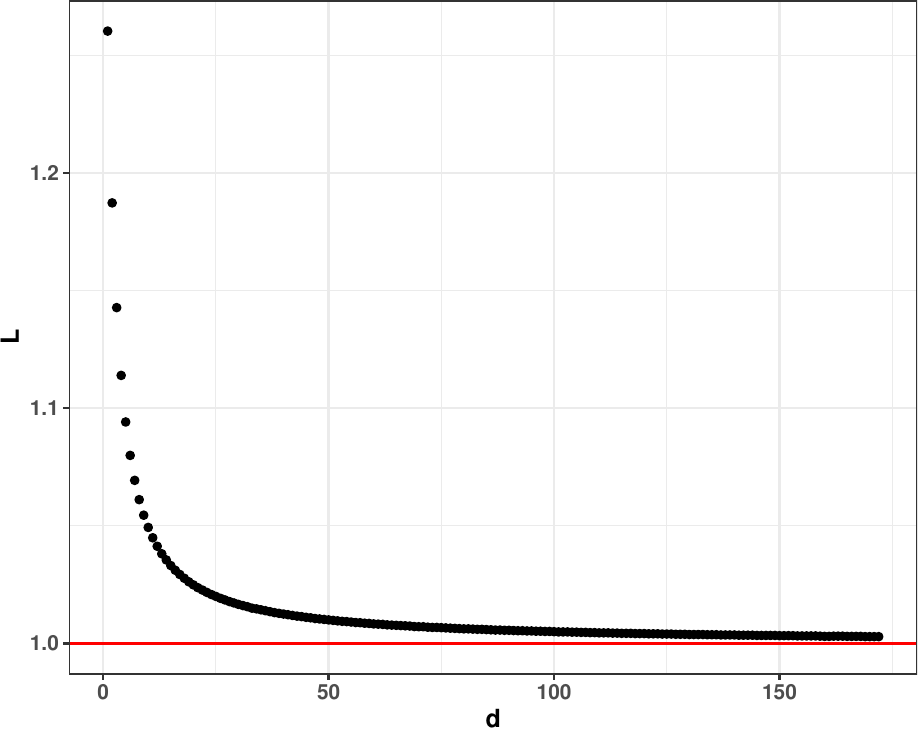}
\caption{\small{ The optimal shifting parameter $L_d$ against the dimension $d$. The optimal radius $c_d=\sqrt{d+L_d}$ approaches $\sqrt{d+1}$ exponentially fast in the sense that the optimal shifting parameter $L_d$ is decreasing and approaches $1$ exponentially fast. At $d=10$, the absolute difference between $L_d$ and 1 is smaller than $0.05$. }}\label{fig-l_apprx_sol} 
\end{figure}

\begin{figure}
\centering
\includegraphics[scale=.5]{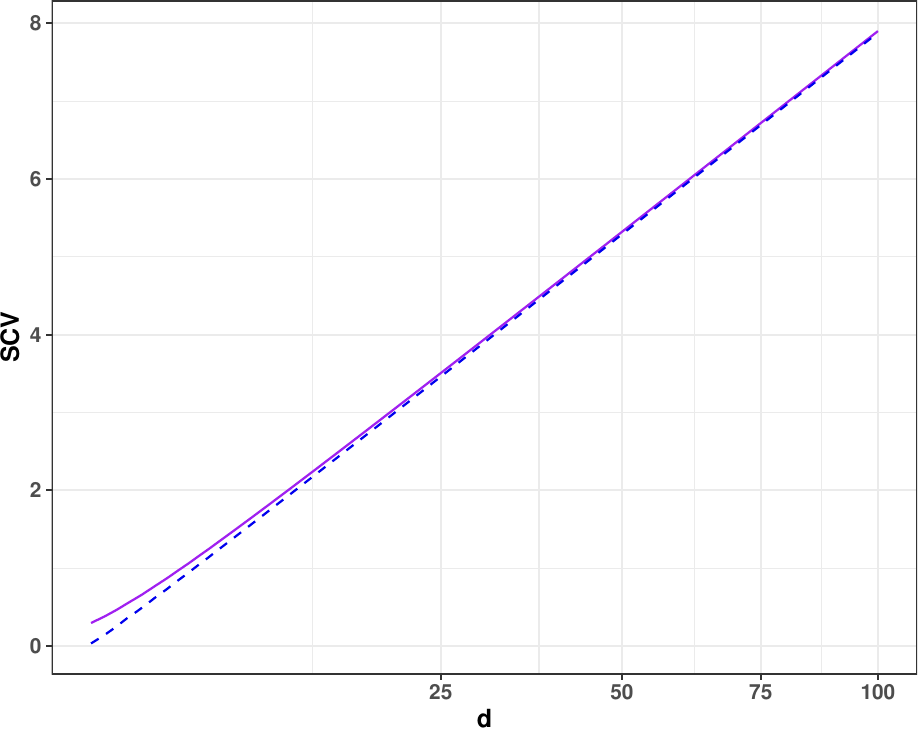}
\caption{\small{$SCV(d,\sqrt{d+1})$ (purple) and the lower bound (blue) plotted against the parameter dimension, $d$, which is plotted on the square root scale. The resulting line shows that the $SCV$ obtained by choosing the radius $\sqrt{d+1}$ is very well approximated by a square root function. It further suggests that this choice results in an $SCV$ (and thus a variance of the THAMES) that is asymptotically optimal, since the $SCV$ approaches its lower bound.}}\label{fig-thames_sqrt_plot_thm1}
\end{figure}

\begin{figure}
    \centering
    \includegraphics[scale=.5]{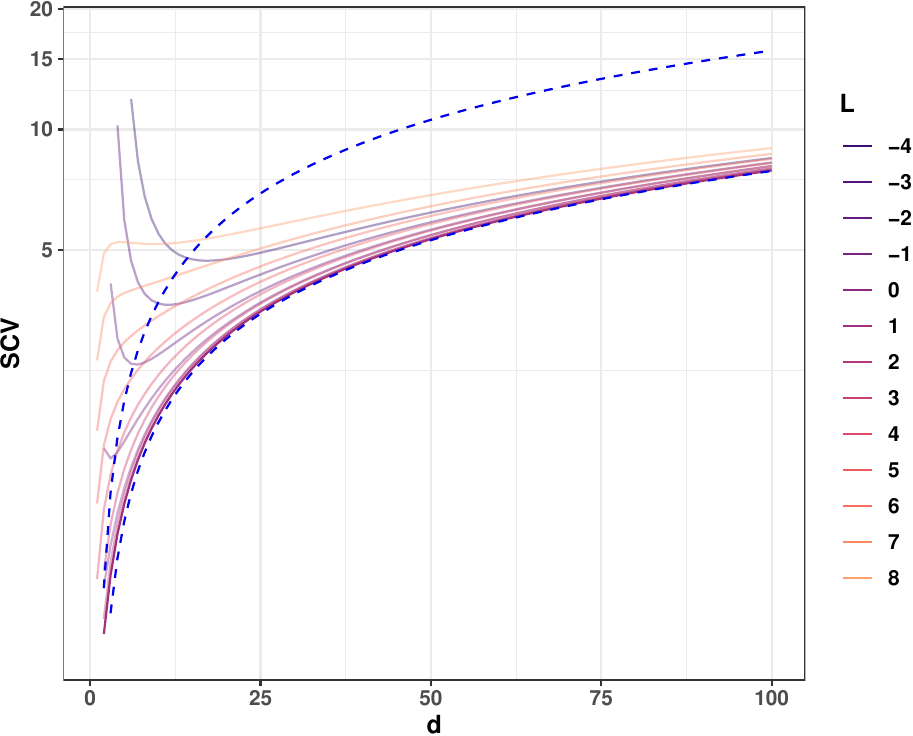}
    \caption{Asymptotic behaviour of $SCV(d,\sqrt{d+L})$ for several values of $L$ ($L=1$ is bright purple) with the lower bound and twice its value included (blue); the scale of the y-axis is logarithmic. All listed values of $L$ eventually lead to an $SCV$ which is at most twice as bad as the optimal $SCV$, even for small $d$. This illustrates that  the variance of the THAMES obtained by choosing one of these suboptimal shifting parameters is asymptotically at most twice as high as the variance obtained by choosing the optimal parameter.}
    \label{fig-thames_scv_and_l}
\end{figure}

\begin{figure}
\centering
\includegraphics[scale=.5]{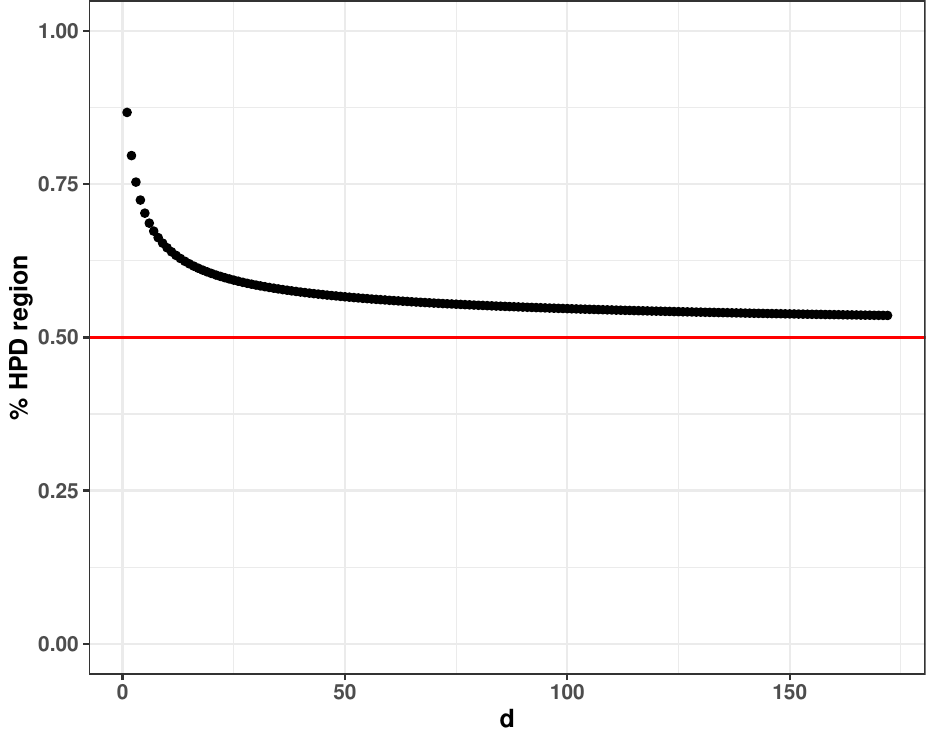}
\caption{\small{Approximate size of the HPD region using the optimal shifting parameter $L_d$. It seems to approach 50\%.\ Thus for a large dimension $d$ and given the assumption of posterior normality, choosing the optimal value for the radius is roughly equivalent to choosing the 50\% HPD region and setting all the terms of the THAMES outside that region to 0.}}\label{fig-thames_hpd_region} 
\end{figure}

\begin{figure}
\centering
\includegraphics[scale=.5]{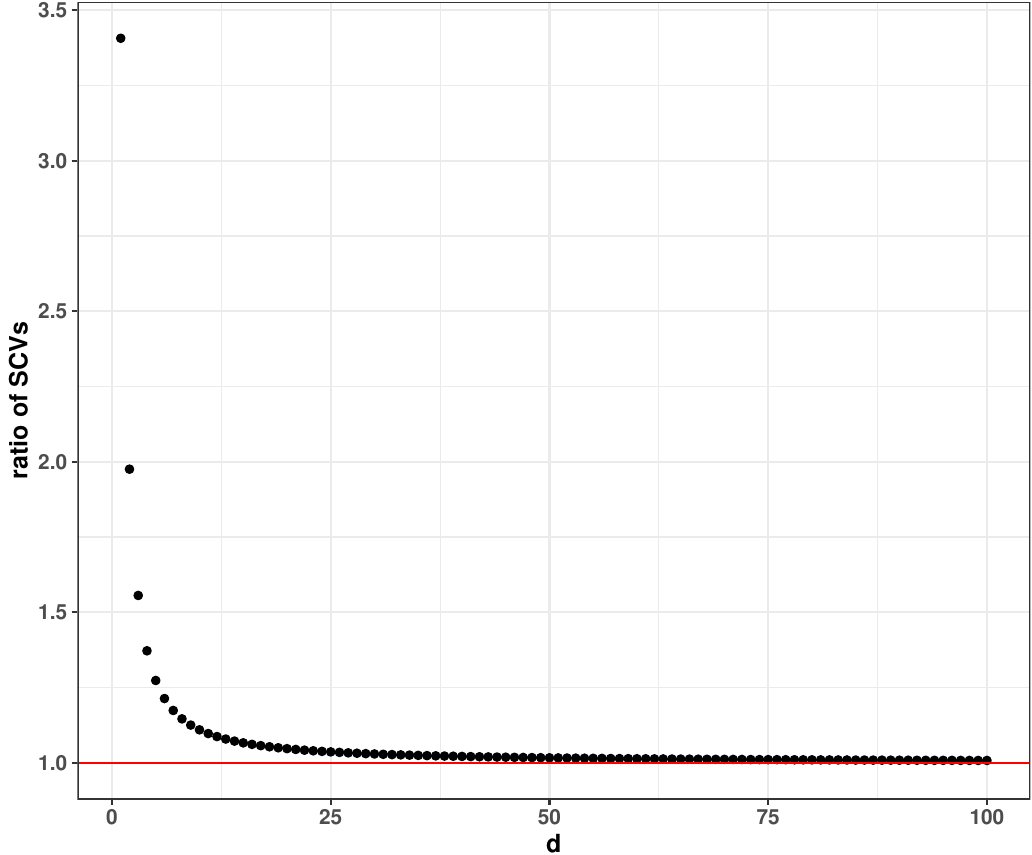}
\caption{\small{The ratio $SCV(d,\sqrt{(\chi^2)^{-1}(50\%,d)})/SCV(d,c_d)$ with varying $d$. It approaches 1 exponentially fast, so the error induced by using this heuristic decreases fast with the parameter dimension, $d$. }}\label{fig-thames_ratio_opt_vs_heuristic} 
\end{figure}

\begin{figure}
\centering
\includegraphics[scale=.6]{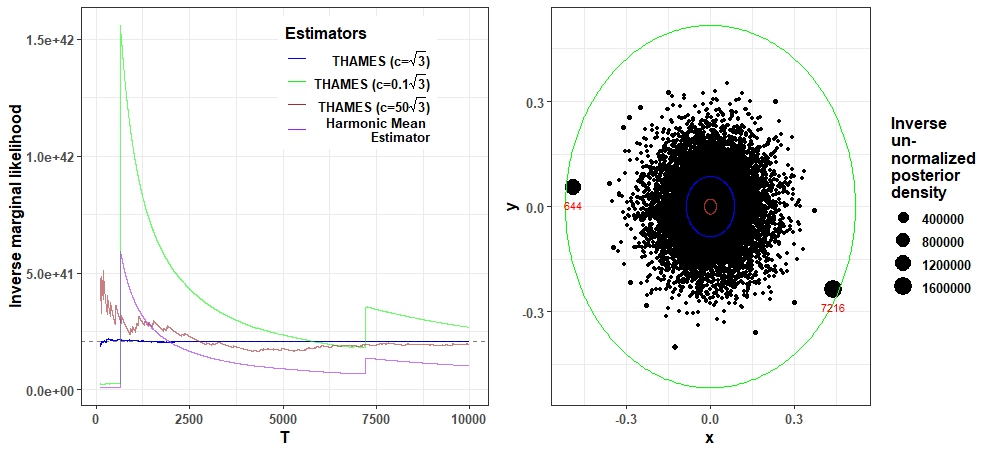} 
\caption{\small{Left: The THAMES calculated by choosing the radii $c=\sqrt{d+1}=\sqrt{3}, c=0.1\sqrt{3}, c=50\sqrt{3}$ (blue, brown, purple) in the two-dimensional case, $d=2$, with the true value $1/Z$ (black) and the Harmonic Mean Estimator (purple). Right: The posterior sample evaluated at the inverse of the unnormalized posterior density and the different ellipses  used to define the THAMES. In this particular case the posterior covariance matrix is the scaled identity matrix, so the ellipses are spheres. The two samples occuring at points 644 and 7216 have a very low likelihood. They cause massive jumps in the Harmonic Mean Estimator when the radius of $A$ is large (green) and are excluded when the radius is equal to $\sqrt{d+1}$ (blue). One can choose a smaller radius (brown), but then too much of the sample is excluded and convergence takes longer.} 
}\label{fig-thames_conv_comp}
\end{figure}


\subsection{THAMES algorithm}\label{sect-algo}

Below is an algorithm for the implementation of THAMES. Procedures for sample splitting, as well as the truncated ellipsoid correction used in the case that the parameter space is constrained have been included. These additions are described in Section \ref{ssect-split} and Section \ref{ssec: bounded_param_problem}, respectively. 

We recommend these additions, but we have also found that in some cases they make almost no difference. For example, sample splitting does not appear to have an impact when the dimension of the parameter space, $d$, is small (Section \ref{ssec:multivariateGaussian} and Section \ref{ssec:bayesianRegression}), while the truncated ellipsoid correction is negligible when the posterior mean is not close to the edge of the posterior support (Section \ref{ssec:mixedeffect} and Section \ref{ssec:dirmultinom}).

\begin{algorithm}
\caption{$\hat{Z}^{-1}$ calculation}
\begin{algorithmic}
\item[]\textbf{Input:} Data $\mathcal{D}$ and posterior samples $(\theta^{(i)})_{i\in\llbracket 1,T\rrbracket}$.
\item[]\textbf{Sample splitting:} Calculate the empirical mean $\hat\theta$ and sample covariance matrix $\hat\Sigma$ based on the first $T/2$ posterior samples $(\theta^{(i)})_{i\in\llbracket 1,T/2\rrbracket}$.
\item[]\textbf{Standardization:} $\tilde{\theta}^{(i)}=(\theta^{(i)}-\hat{\theta})\hat{\Sigma}^{-1/2}$ for $i\in\llbracket T/2+1,T\rrbracket$. 
\item[]\textbf{Truncation subset:} $\mathcal{S} =\{i :  \|\tilde{\theta}^{(i)}\|_2^2 < d+1 \}$.
\item[]\textbf{Calculate THAMES estimator:} 
\begin{equation*}
 \hat{Z}^{-1}= \frac{1}{T/2}\sum_{i=T/2+1}^{T} \frac{h(\theta^{(i)})}{L(\theta^{(i)})\pi(\theta^{(i)})},
\end{equation*}
\qquad where $h(\theta^{(i)})= 1/V(A)$ if $i\in \mathcal{S}$ and 0 otherwise, with \\
\qquad $V(A)=\sqrt{|\hat{\Sigma}|}\pi^{d/2}(d+1)^{d/2}/\Gamma(\frac{d}{2}+1) $ and $A = \{\theta : (\theta-\hat\theta)^T\hat\Sigma^{-1}(\theta-\hat\theta) < d+1\}$. 
\hspace{1cm} \If{the posterior support $\text{supp}(\theta|\mathcal{D})$ is constrained} 
    \State Simulate the sample $\nu^{(1)},\dots,\nu^{(N)}$ from the uniform distribution on $A$.
    \State Approximate the volume ratio $V(A\cap \text{supp}(\theta|\mathcal{D}))/V(A)$ via the Monte Carlo estimator
    \begin{align*}
        \hat{R}=\frac{1}{N}\sum^N_{i=1}\mathbbm{1}_{\{\theta|\pi(\theta)L(\theta)>0\}}(\nu^{(i)}).
    \end{align*}
    \State Assign $\hat{Z}^{-1} \leftarrow \hat{R}^{-1}\hat{Z}^{-1}$.
\hspace{1cm} \EndIf
\item[]\textbf{Output:} THAMES estimator $\hat{Z}^{-1}$.
\end{algorithmic} 
\end{algorithm}

\subsubsection{Sample splitting}\label{ssect-split}

The theoretical guarantees established in Section \ref{ssec:copt} operate under the assumption of an oracle ellipsoid $A_{or}$. In particular, this means that the ellipsoid determining the THAMES estimator $\hat{Z}^{-1}$ is defined independently of the posterior samples $(\theta^{(i)})_{i\in\llbracket 1,T\rrbracket}$. In practice, we find that estimating $A$ and $Z^{-1}$ simultaneously using the same posterior sample can induce bias in $\hat{Z}^{-1}$ when the parameter space is high-dimensional. As such, we implement a sample splitting procedure that involves estimating $A$ and $Z^{-1}$ using separate posterior draws. Specifically, we first estimate the posterior mean and covariance matrix via the empirical mean $\hat\theta$ and sample covariance $\hat\Sigma$ using the first $T/2$ posterior samples $(\theta^{(i)})_{i\in\llbracket 1,T/2\rrbracket}$. Defining $A$ as in Equation \eqref{eq-A} based on $\hat\theta$ and $\hat\Sigma$, we then calculate the THAMES estimator $\hat{Z}^{-1}$ using the last $T/2$ posterior samples $(\theta^{(i)})_{i\in\llbracket T/2+1,T\rrbracket}$.

\subsubsection{Correcting for the presence of constrained parameters}\label{ssec: bounded_param_problem}

Whenever the posterior support of the parameters is not equal to $\mathbbm{R}^d$, for example when the parameters are variances or probabilities, it is possible that our choice of $h$ in Equation \eqref{eq-risest}, the pdf of the uniform distribution on $A$, is not correctly normalized. This is due to the fact that $A$ is not necessarily a subset of the posterior support and thus $h$ is not a pdf over this space. 

In this case, the expectation of the THAMES is distorted by a multiplicative constant:\begin{align}\label{eq-bounded_param_error}
    E_{\theta}[\hat{Z}^{-1}|\mathcal{D}]=E_{\theta}  \left[ \frac{h(\theta)}{L(\theta) \pi(\theta)}
 \bigg| \mathcal{D} \right]=Z^{-1}\cdot \frac{V(A\ \cap \ \text{supp}(\theta|\mathcal{D}))}{V(A)}=:Z^{-1}R ,
\end{align}
where $V(A\ \cap\  \text{supp}(\theta|\mathcal{D}))$ denotes the volume of the intersection between $A$ and the posterior support. One way to deal with this problem is to transform the parameter space, e.g., by setting $\vartheta:=\log(\theta)$ if $\theta$ is a variance parameter. One can then continue with marginal likelihood estimation on $\vartheta$, using the transformed prior distribution. In this case, it is of course important to include the Jacobian of the transformation when computing the prior density.

Another way is to adjust for the bias by calculating the ratio of these volumes, $R$, using a simple Monte Carlo approximation: We simulate $\nu^{(1)},\dots,\nu^{(N)}\stackrel{\text{i.i.d.}}\sim\text{Unif}(A),N\in\mathbb{N}$ and calculate 
\begin{align}
    \hat{R}:=\frac{1}{N}\sum^N_{i=1}\mathbbm{1}_{\text{supp}(\theta|\mathcal{D})}(\nu^{(i)})=\frac{1}{N}\sum^N_{i=1}\mathbbm{1}_{\{\theta|\pi(\theta)L(\theta)>0\}}(\nu^{(i)}).
\label{eq:rhat}
\end{align}
Given $A$, this is an unbiased and consistent estimator of $R$ by the law of large numbers. The bias-adjusted THAMES is then \begin{align}
    \hat{Z}^{-1}_{adj}=\hat{R}^{-1}\hat{Z}^{-1}.
\end{align}

The problem of the parameter space being constrained is common not only for the THAMES, but for reciprocal sampling estimators in general. It has for example been addressed by \cite{Hajargasht2018} and \cite{Sims2008}. \cite{Hajargasht2018} used variational Bayes techniques and showed that these ensure that the support of the chosen $h$ is a subset of the posterior support, under mild conditions. \cite{Sims2008} used an ellipsoidal density truncated on a subset of the joint support, $\Theta_U:=\{\theta\pi(\theta)L(\theta)>U\}$, where $U>0$. Since the support of $\pi(\theta)L(\theta)$ is equal to the posterior support, our truncation set is similar to the one chosen in \cite{Sims2008}, except that we set $U=0$.

The adjustment is usually very small: The problem arises only when the posterior mean is close enough to the edge of the parameter space. The edge of the parameter space often inherently indicates a priori unlikely values. For this reason it is also rare that the data indicates posterior parameters being close to the edge. Thus the ratio between the volumes is close to one and the variance of $\hat{R}$ is small. In fact, the adjustment did not have any sizeable impact on any of the examples simulated in Section \ref{sect-examples}. This may not be the case, however, if the actual data generating mechanism are very different from the model being considered.
In this case, it can in practice happen that the posterior mean is indeed very close to the edge. We show one example of this in Section \ref{ssec:dirmultinom}.

In either case we have found that a small number of simulations, around $N=100$,  is usually enough. Confidence intervals obtained from the fact that $\hat{R}$ is asymptotically normal can be used to check whether the variance of $\hat{R}$ is large. In this case $N$ should be increased to yield a more precise approximation.

\section{Examples}
\label{sect-examples}

 We now describe several simulated and real data examples to assess the THAMES estimator. In Sections \ref{ssec:multivariateGaussian}, \ref{ssec:bayesianRegression}, and \ref{ssec:dirmultinom}, three statistical models, for which exact expressions of the marginal likelihood are derived, are considered. This allows us to compare the THAMES estimated values to the exact ones for evaluation. In Section \ref{ssec:mixedeffect}, we consider a real data example with models for which no analytical expressions for the marginal likelihood are available and where there is a  need for reliable estimators. We compare our estimator to bridge sampling, which is more complicated than THAMES but is known to have perform well \citep{MengWong1996,gronau2020}.
 
\subsection{Multivariate Gaussian data}\label{ssec:multivariateGaussian}
We first consider the case where data $Y_i, i=1,\dots,n$ are drawn independently
from a multivariate normal distribution:
\begin{eqnarray*}
Y_i|\mu & \stackrel{\rm iid}{\sim} & {\rm MVN}_d(\mu,  I_d), \;\; i=1,\ldots, n,
\end{eqnarray*}
along with a prior distribution on the mean vector $\mu$:
\begin{equation*}
    p(\mu)= {\rm MVN}_d(\mu; 0_d, s_0 I_d),
\end{equation*}
with $s_0 > 0$. As shown in the Appendix, the posterior distribution of the mean vector $\mu$ given the data $\mathcal{D}=\{y_1, \dots, y_n\}$ is given by:
\begin{equation}\label{eq:postMultiGauss}
    p(\mu|\mathcal{D}) = {\rm MVN}_d(\mu; m_n,s_n I_d),
\end{equation}
where $m_n=n\bar{y}/(n+1/s_0)$, $\bar{y}=(1/n)\sum_{i=1}^n y_i$, and $s_n=1/(n +1/s_0)$. 

Interestingly, while the observations $(Y_i)_i$ are independent given the vector $\mu$, they are not independent marginally, and the marginal likelihood does not take a product form over marginal terms in $i$. Conversely, thanks to the isotropic Gaussian prior distribution which is considered for $\mu$, where the $(\mu_j)_j$ are all iid, not only are the vectors $(Y_{.j})_j$ independent given $\mu$, they are also independent marginally. From such a key property, we prove in Proposition \ref{appendix:multivariateGaussian} of the Appendix that the marginal likelihood of the model can be written analytically as:
\begin{equation}\label{eq:exactMultiGauss}
    p(\mathcal{D}) = \prod_{j=1}^d {\rm MVN}_n (y_{.j}; 0_n, s_0 1_n 1_n^{\intercal} +   I_n),
\end{equation}
where $y_{.j} \in \mathbb{R}^n$ is the vector of all observations for variable $j$ such that $[y_{.j}]_i = y_{ij}$ and $1_n$ is the vector of 1 in $\mathbb{R}^n$.

\subsubsection{Assessing the precision of the THAMES estimator as a function of $T$}

We first considered the univariate case $d=1$. Thus, we simulated a unique sample of 
size $n=20$ with $\mu=2$. Moreover, we set $s_0 = 1$, for illustration. Note that other choices for $s_0$ led to similar conclusions regarding the quality of the estimation. Figure \ref{fig-univgaussian} shows the THAMES estimated
values for the log marginal likelihood, for $T=5, 1005, 2005, \ldots, 9005$ samples of the posterior distribution (Equation \eqref{eq:postMultiGauss}). Confidence intervals as well as the exact value of the log marginal likelihood computed using Equation \eqref{eq:exactMultiGauss} are also reported. It can be seen that the estimate converges to the correct value and that the confidence intervals contain the true value in all cases, even for $T=5$ only.

\begin{figure}[t]
    \centering
    \includegraphics[scale=0.5]{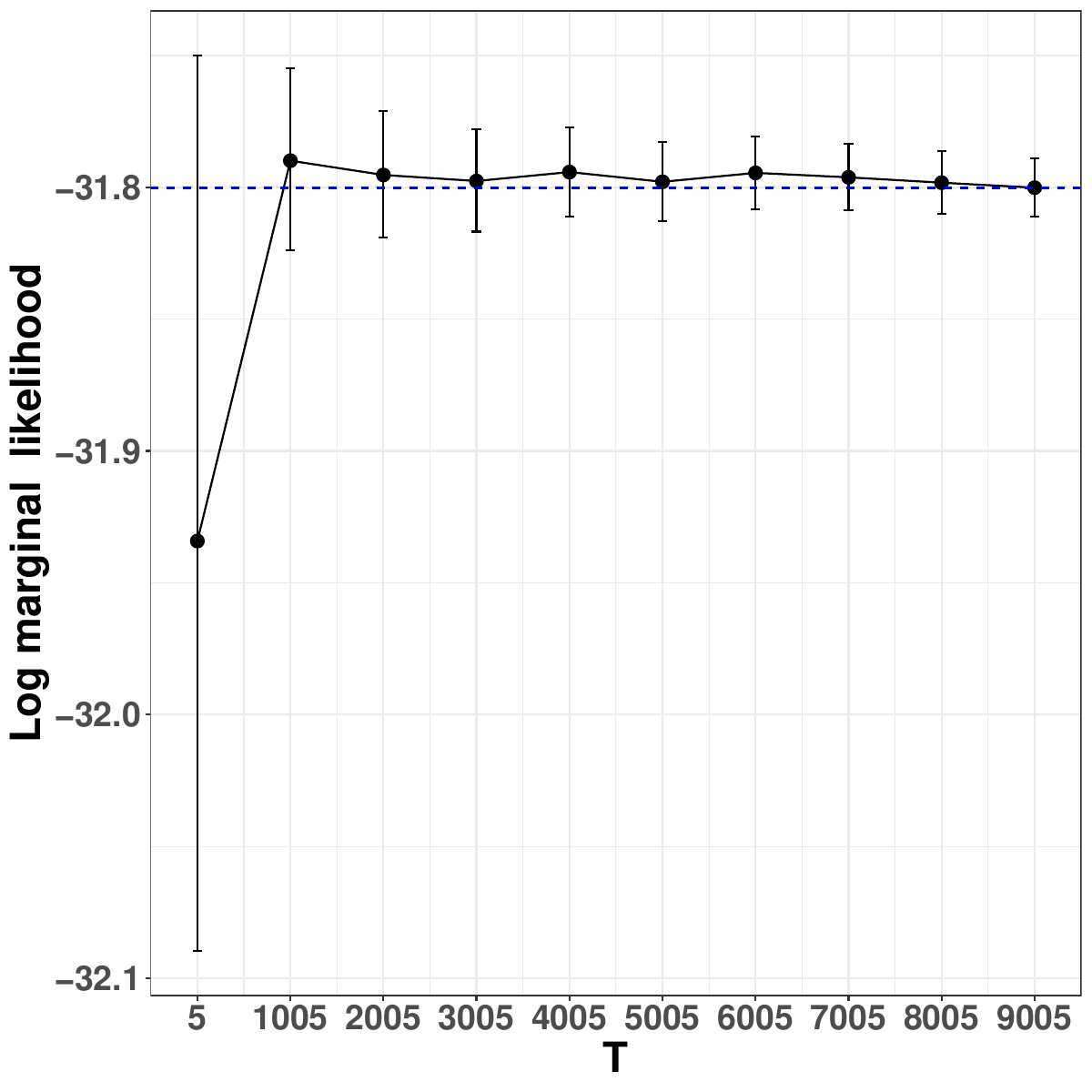}
    \caption{Estimation of the log marginal likelihood using THAMES for a unique univariate Gaussian sample with
$n=20$ as a function of $T$, the number of (cumulative) samples from the posterior distribution. The black dots indicate the values of the THAMES estimator of the log marginal likelihood. The vertical lines represent 95\% confidence intervals,
and the dashed blue line represents the exact value computed using Equation \eqref{eq:exactMultiGauss}.}
    \label{fig-univgaussian}
\end{figure}

\subsubsection{Assessing the precision of the THAMES estimator as a function of $d$}

For this second set of experiments, we considered different values of $d$, and aimed at testing the robustness of the THAMES approach on multiple data sets, with increasing dimensionality. Thus, for each $d$, we generated 50 different data sets of size $n=20$ using the multivariate Gaussian model. In practice, we set the true value of $\mu$ to $2$, for all its components. Again, the prior parameter $s_0$ was set to $1$ and similar conclusions were drawn for other values. Moreover, the value of $T$ was set to $10,000$ for all the experiments. 

We also used this example to assess  sample splitting procedure for the posterior samples, as proposed in Section \ref{ssect-split}. The results  are given in Figure \ref{fig-univgaussian}. On the figure on the left, where \emph{no}  sample splitting of the posterior samples is employed to compute THAMES, we observe that a bias appears as the dimensionality of the model considered increases, and the log marginal likelihood tends to be slightly underestimated. As illustrated by the figure on the right hand side, this bias is primarily related to the estimation of the posterior covariance matrix, and not to the THAMES estimation itself. Indeed, focusing on this figure on the right,  we note that if the exact expression  of the posterior covariance matrix given in Equation (\ref{eq:postMultiGauss}) is used to compute THAMES, then while the variance of the estimator increases with $d$, we do not observe any bias. Crucially, if the   sample splitting of the posterior samples is employed to compute THAMES, then again, we do not observe any bias. 

Overall, we found that the  sample splitting procedure of the posterior samples was not necessary to compute THAMES for low values of $d$. The estimated values are indeed particularly close to the exact ones. However, for large values of $d$, we recommend using the  sample splitting procedure to remove the bias.

\begin{figure}
    \centering
   \begin{tabular}{ccc}
       \includegraphics[width=0.3\textwidth]{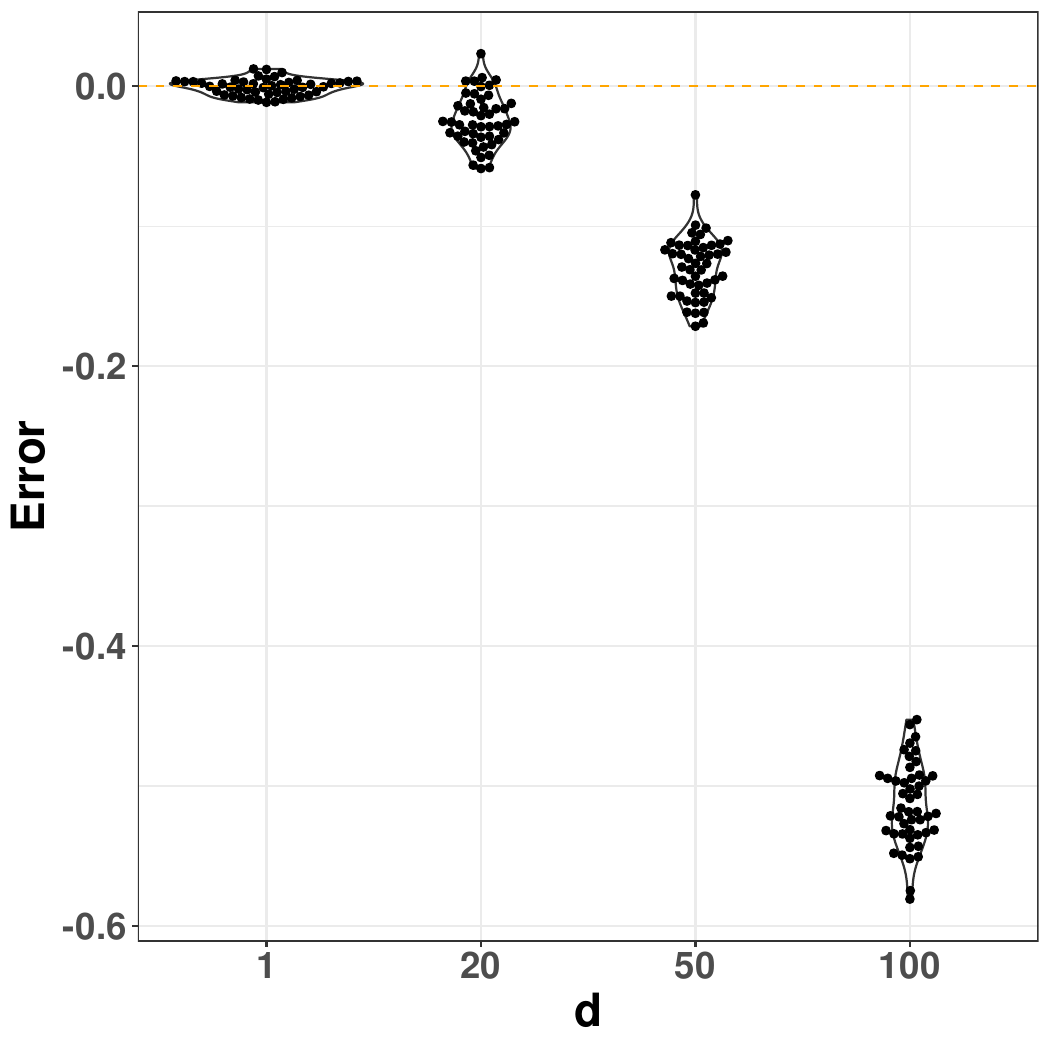} & \includegraphics[width=0.3\textwidth]{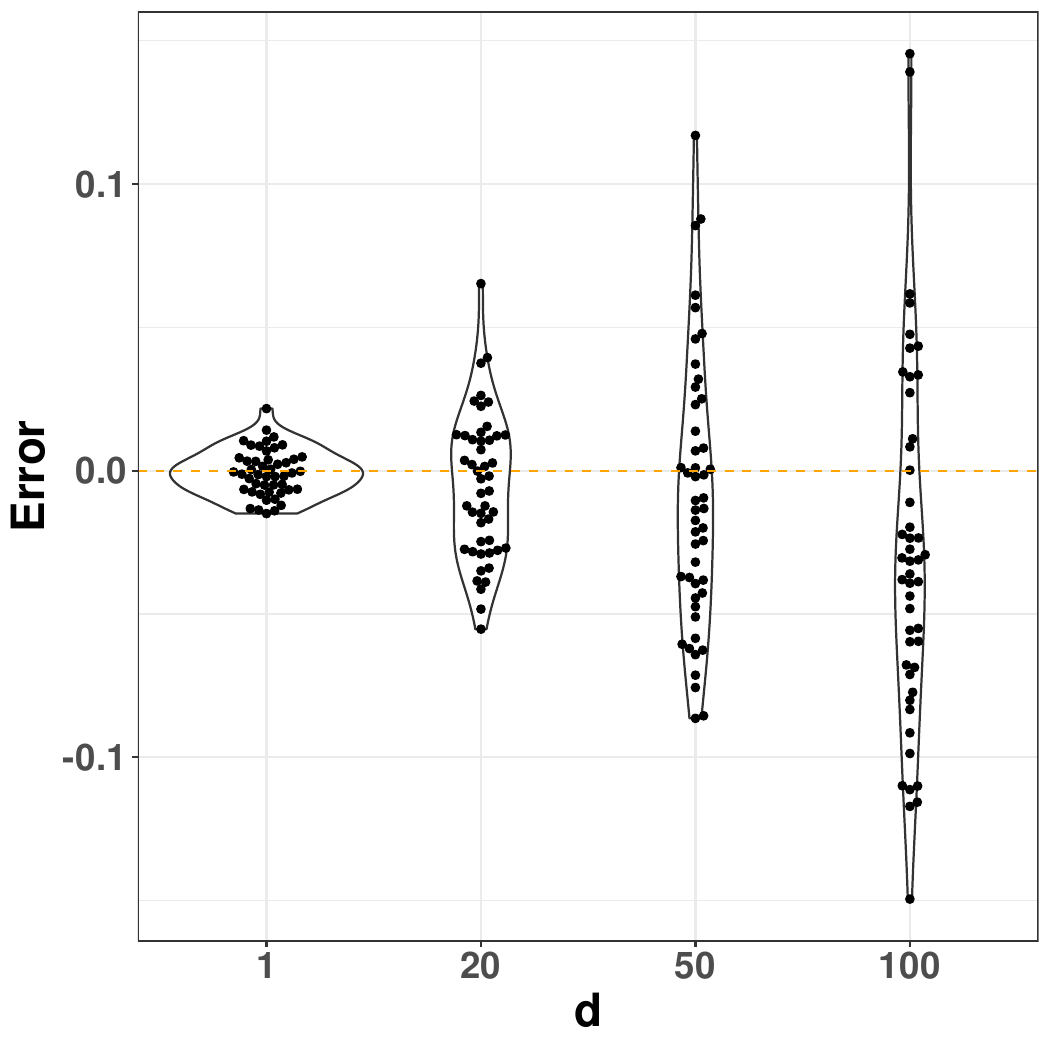}   & \includegraphics[width=0.3\textwidth]{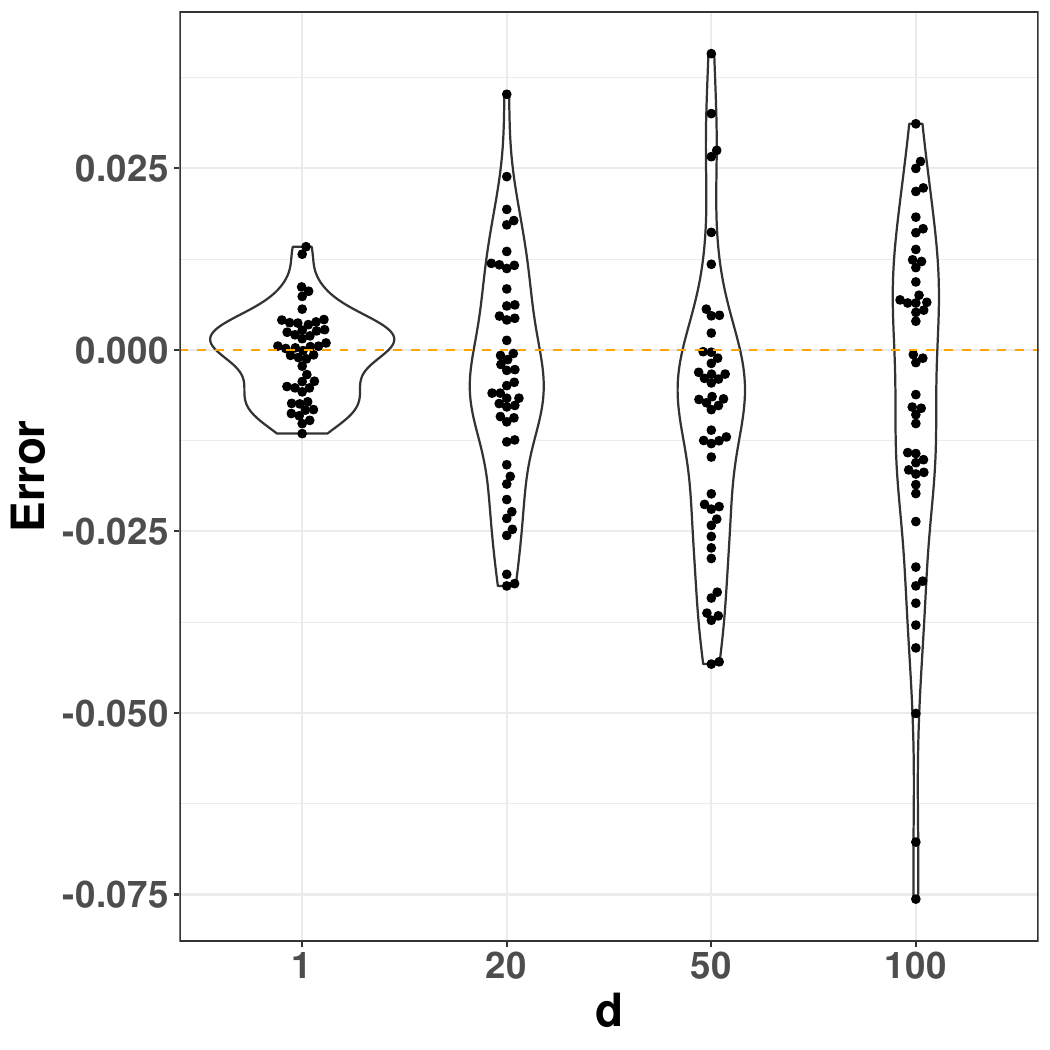}
    \end{tabular}
   
    \caption{Difference between the estimated log marginal likelihood using THAMES  and the true log marginal likelihood for a multivariate Gaussian model with
 $n=20$ and $T=10000$. The procedure is repeated on 50 different data sets for each $d$. Left: the THAMES approach \emph{with no}  sample splitting of the posterior samples. Middle: the THAMES approach \emph{with}  sample splitting of the posterior samples. Right: the results provided correspond to the case where the exact expression of the posterior covariance matrix given in Equation (\ref{eq:postMultiGauss}) is used to compute THAMES.}
    \label{fig-univgaussian}
\end{figure}

\subsection{Bayesian Regression}\label{ssec:bayesianRegression}

We consider a  data set $(x_i, Y_i), i=1,\dots,n$ to train a linear regression model of the form
\begin{equation*}
    Y_i | x_i, \beta \sim \mathcal{N}(x_i^{\intercal} \beta, \sigma^2), i=1,\dots, n.
\end{equation*}
Note that we assume the noise variance $\sigma^2$ to be known for the sake of the illustration.  Indeed, as the goal in this section is to assess the quality of the estimator we propose, an exact expression for the marginal likelihood $Z$ is needed, for comparison. If a prior distribution such as a gamma distribution for instance is chosen for $\sigma^2$, then, while the methodology we propose in this paper can be used directly for estimation, no analytical expression for $Z$ is not available, making the assessment infeasible.

Denoting $Y \in \mathbbm{R}^n$, the vector of target variables $Y_i$, and $X \in \mathcal{M}_{n \times d}(\mathbbm{R})$ the design matrix where the input vectors $x_i \in \mathbbm{R}^d$ are stacked as row vectors, the linear regression model becomes: 
\begin{equation*}
Y| X, \beta \sim {\rm MVN}_n(X\beta, \sigma^2 I_n).
\end{equation*}
We rely on a centered isotropic Gaussian prior distribution for the regression vector $\beta$:
\begin{equation*}
p(\beta) = {\rm MVN}_d(\beta; 0_d, I_d / \alpha),
\end{equation*}
with $\alpha > 0$. This framework was largely covered by the PhD thesis of D. MacKay through the development of the so called \emph{evidence} procedure \citep{mackay1992bayesian}.
In this context, the posterior distribution over $\beta$, given the training data set $\mathcal{D} = \{(x_1, y_1), \dots, (x_n, y_n)\}$ is tractable:
\begin{equation*}
    p(\beta| \mathcal{D}) = {\rm MVN}_d(\beta; m_n, \Sigma_n),
\end{equation*}
with
\begin{equation*}
    \Sigma_n^{-1} = \frac{X^{T}X}{\sigma^2} + \alpha I_d,
\end{equation*}
and
\begin{equation*}
    m_n = (X^{\intercal}X + \alpha \sigma^2 I_d)^{-1} X^{\intercal} \textbf{y},
\end{equation*}
where $\textbf{y} \in \mathbb{R}^n$ the \emph{observed} vector of target variables associated to $Y$. Moreover, the marginal likelihood also has an analytical expression:
\begin{equation*}
    p(\textbf{y}|X) = {\rm MVN}_n(\textbf{y}; 0_n, \frac{X X^{\intercal}}{\alpha} + \sigma^2 I_n).
\end{equation*}
The two corresponding proofs are given in Proposition \ref{prop:linreg} of the Appendix. 

The data for this example are described by \citet{Hastie&2009} and come from a study by \citet{Stamey&1989}. They examined the correlation between the level of prostate-specific antigen (lpsa) and eight clinical measures in men who were about to receive a radical prostatectomy. The variables are log cancer volume (lcavol), log prostate weight (lweight), age, log of the amount of benign prostatic hyperplasia (lbph), seminal vesicle invasion (svi), log of capsular penetration (lcp), Gleason score (gleason), and percent of Gleason scores 4 or 5 (pgg45). The target variable is the level of prostate-specific antigen (lpsa). Note that in our experiments, we set $\alpha=1/2$, but $\alpha$ could also be estimated with the evidence procedure. Other choices for $\alpha$ led to similar conclusions regarding the quality of the estimation. 



Seven different regression models $M_2$, $M_3$, $\dots$, $M_8$, each with a different number of  selected variables, ranging from 2 to 8, are considered for illustration. The variables are added in the order given above. Thus, $M_2$ is made of variables lcavol and lweight, while Model $M_3$ considers the variables  lcavol, lweight, as well as age for prediction. Finally, model $M_8$ takes all input variables into account. Figure \ref{fig-prostate.marglik} shows the THAMES of the log marginal likelihood for different number of samples from the posterior distribution in $\beta$, for the different models, as well as the approximate confidence intervals.  We did not use sample splitting to estimate $A$ in this example, as the dimension of the parameter vector was relatively low. 

There was no noticeable bias in the results. Indeed, it can be seen that the estimators converge
rapidly to the correct value and that the intervals cover the correct values
in most cases, even when the number of samples used is small, for all models investigated. For all models, the THAMES estimator is particularly accurate for 1000 samples of the posterior only. While the main goal of this section is to illustrate the precision of our estimation strategy for a series of models, we can also report that the model with the highest marginal likelihood, for this data set, is Model $M_2$. In other words, the variables lcavol and lweight are seen as key for the prediction of the level of prostate-specific antigen.

\begin{figure}[htbp!]
    \centering
    \includegraphics[width=0.8\textwidth]{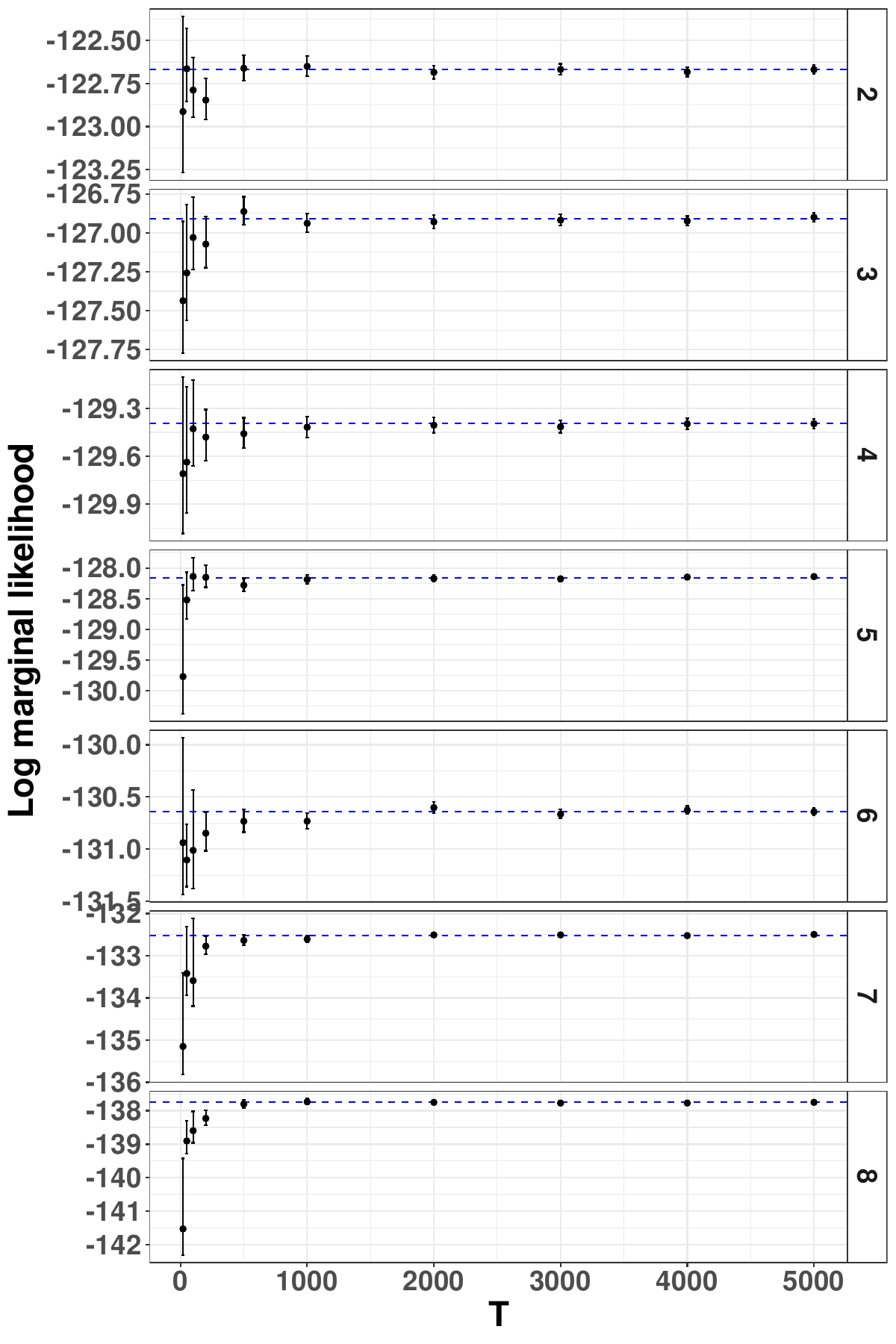}
    \caption{The black dots indicate the values of the THAMES estimator, for different models and number of samples from the posterior distribution, in the prostate dataset. The approximate confidence intervals of the estimator are also indicated in black. The exact values of the log marginal likelihood are shown by the dashed blue lines. }
    \label{fig-prostate.marglik}
\end{figure}

\subsection{Dirichlet-multinomial model}\label{ssec:dirmultinom}

Extensions of the Dirichlet-multinomial model are widely used in the context of topic modelling (see, e.g., \cite{Blei2003LDA}) The expression for the marginal likelihood in this model is known, as in the previous two sections. This allows us to assess the performance of our estimator in another simulation study, in a non-Gaussian context. 

A simulation study in this setting is useful for two reasons: First, this is a high-dimensional setting in which the posterior distribution of the parameters is highly non-Gaussian. In fact the parameter space is bounded. This allows us to assess how well the THAMES performs in a very different setting and also lets us assess how much of an impact the correction for a bounded parameter space from Section \ref{ssec: bounded_param_problem} has.

Second, there do exist similar models to this one for which the marginal likelihood is not tractable,  e.g. \citet{Blei2007CTM}). These models are therefore a possible application of the THAMES. The simulation study might give an idea of how well the THAMES would perform in these applications.

The Dirichlet-multinomial model is defined as follows: Each data point $Y_i\in\{0,\dots,l\}^K$ is drawn from a multinomial distribution given a Dirichlet-distributed random variable $\mu$:\begin{align*}
\mu&\sim \text{Dirichlet}(\mu;(a_0,\dots, a_0)), \\ Y_i|\mu&\stackrel{\text{i.i.d.}}\sim\mathcal{M}(l,\mu),\quad \forall i=1,\dots, n . \end{align*} Here, $\mu$ is positive and $K$-dimensional with components summing to 1. The covariance matrix of $\mu$ is thus necessarily singular. As noted in Remark \ref{rem-sigma_mustbe_pos_def}, the THAMES needs to be used on posterior simulations from the subspace of $\mathbb{R}^K$ on which a density is defined. In this case, this is $\mathbb{R}^{K-1}=:\mathbb{R}^d$. The prior density is thus \begin{align*}
    \pi(\mu_1,\dots,\mu_d)=\text{Dirichlet}\left(\mu_1,\dots,\mu_d,1-\sum^d_{j=1}\mu_j;(a_0,\dots,a_0)\right).
\end{align*} The posterior support is the $d$-dimensional simplex \begin{align*}
    \left\{\mu\in\mathbbm{R}^d\ |\ \sum^d_{j=1}\mu_j\leq1,\mu_1,\dots,\mu_d>0\right\}.
\end{align*}The posterior distrib ution given the data $\mathcal{D}=\{y_1,\dots,y_n\}$ is tractable:\begin{align*}
    p(\mu_1,\dots,\mu_d|\mathcal{D})&=\text{Dirichlet}\left(\mu_1,\dots,\mu_d,1-\sum^d_{j=1}\mu_j;\alpha_1,\dots,\alpha_K\right) ,\\ \alpha_j&=a_0 + \sum_{i=1}^ny_{ij} .
\end{align*}
The marginal likelihood is thus also tractable, using Bayes's theorem.

\subsubsection{Results}

The marginal likelihood was estimated in the setting $(n,l,T,a_0)=(400,150,10000,1)$ with $d$ varying between 1,20, 50 and 100. The quantities $n$ and $l$ were intentionally chosen to be large, since this model has very high dimensional applications. For example, \cite{Blei2003LDA} used a data set with 8000 documents, $n=15,818$ words and used up to $K=200$ different topics. 

We considered two different approaches for the simulation procedure. In the first, a fixed value of $\mu$ was determined using the Dirichlet distribution with parameters $(a_0,\dots,a_0)$ for each value of $d$. For each $\mu$, 50 data sets were simulated using the multinomial distribution with parameters $l$ and $\mu$. The second approach was to not simulate $\mu$, but fix it to be the probability vector of the uniform distribution, i.e. $\mu=(1/K,\dots,1/K)$ and repeat the same simulation procedure just described. The results of both of these approaches are shown in Figure \ref{fig-simple_multinomial_3_plots}. The posterior sample was split: the first $T/2$ parameter values simulated from the posterior were used  to estimate the covariance matrix and mean, while the remaining values were  used to calculate the THAMES.

Overall, the THAMES performed well in all the experiments, and the impact of the bounded parameter correction was small in this case. The performance of the THAMES in the second setting, with fixed $\mu$, was noticeably better than its performance in the first, with stochastic $\mu$. In fact, the third plot in Figure \ref{fig-simple_multinomial_3_plots} is quite similar to the third plot in Figure \ref{fig-univgaussian}, which shows the Gaussian case. Presumably this can be explained by Theorem \ref{thm-n_limit_approx}: it seems that the size of the data set is large enough that the posterior is well approximated by a multivariate Gaussian, for which the parameter of the THAMES is optimal.

The results from the approach in which $\mu$ was simulated from the Dirichlet (the left plot in Figure \ref{fig-simple_multinomial_3_plots}) are less symmetric and show larger errors on average than the results obtained in which $\mu$  was fixed to be the probability vector of the discrete uniform. This is probably because probability vectors simulated from the Dirichlet are more skewed, with some of their elements being close to the edge of the parameter space. For this reason convergence of the posterior to the multivariate Gaussian presumably takes longer. The results are still similar to those the baseline setting in which $\mu$ was fixed: the errors increase with the number of dimensions, roughly at rate $\sqrt{d}$, as predicted by Theorem \ref{thm-limit_behav_SCV}. 

As mentioned, probability vectors simulated from the Dirichlet can often have elements close to the edge of the parameter space. The bias correction from Section \ref{ssec: bounded_param_problem} was thus used in the first approach, although this correction only had a small impact: The values of the log marginal likelihood estimation where changed by at most 0.007 (plot 2 in Figure \ref{fig-simple_multinomial_3_plots}). The Bias-correction was 0 each time for the $``$base-line$"$ approach. This makes sense, as none of the elements of the probability vector of the discrete uniform are close to the boundary of the simplex. 

We have simulated the THAMES in a high dimensional, non-Gaussian setting. Still, the results obtained showed similarity to those obtained in another setting in which the posterior was exactly Gaussian. This shows a real strength of the THAMES: While a large sample size hugely increases the size of the marginal likelihood and could thus be thought to increase the difficulty of the estimation procedure, it has in our experience \textbf{improved} the performance of the THAMES. The THAMES could thus be a candidate for very high dimensional estimation procedures in which both the dimension of the parameter space and the sample size are large.

\begin{figure}[t]
    \centering
    \includegraphics[scale=.5]{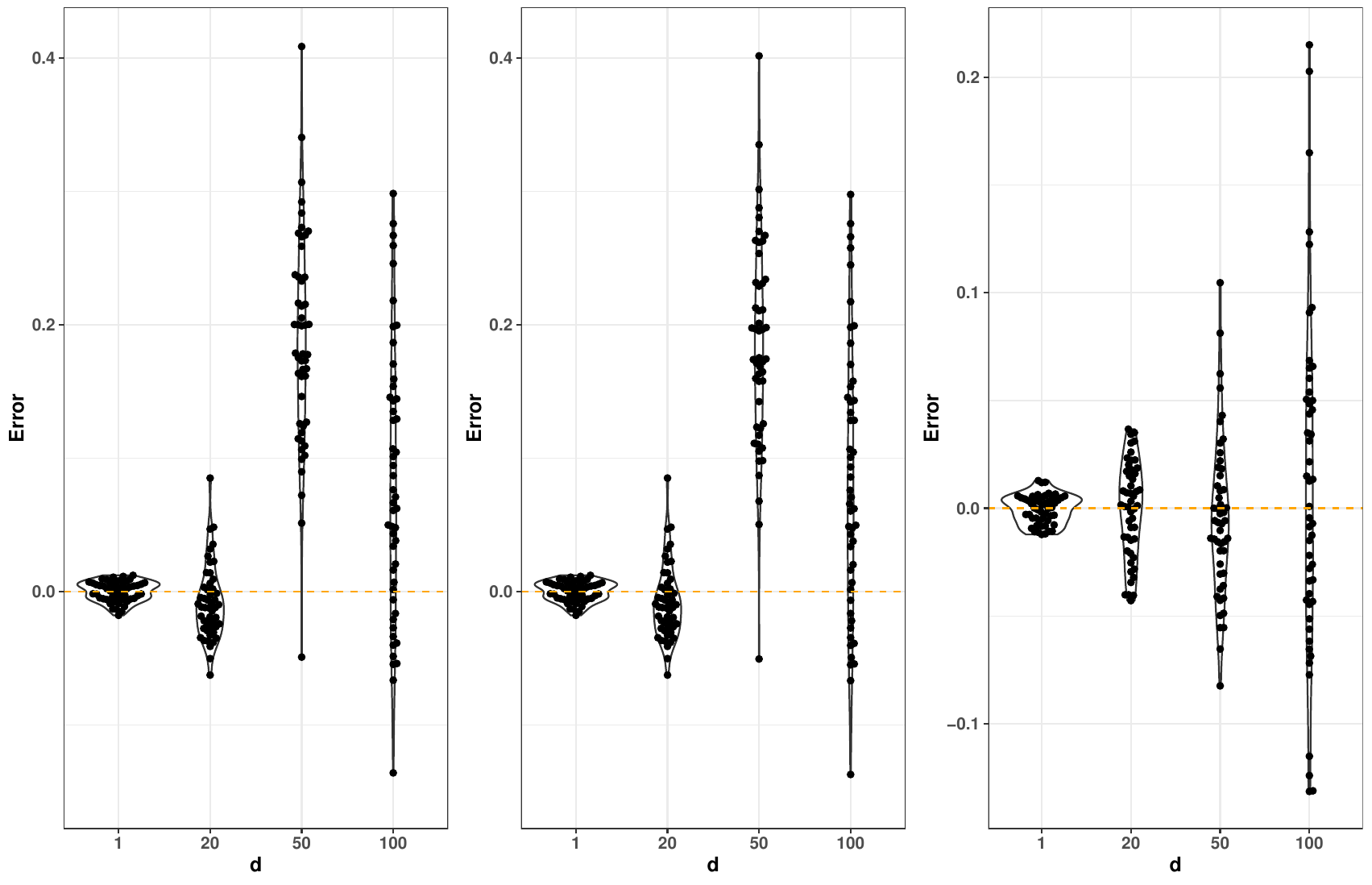}
    \caption{Difference between the estimated log marginal likelihood using THAMES and the true log marginal likelihood for a Dirichlet-multinomial model with
$(n,l,T,a_0)=(400,150,10000,1)$. The procedure is repeated on 50 different data sets for each $d$.  Left: The results when the data sets were simulated using $\mu$ from the Dirichlet-multinomial model, without the bounded parameter space bias correction from Section \ref{ssec: bounded_param_problem}. Center: The same results, but with the bounded parameter space correction. It has a small impact: Some points from the data sets simulated for $d=50$ are shifted by about $0.007$. Right: The model when fixing $\mu$ to be the probability vector of the discrete uniform. It shows better results, presumably because the probability vector of the discrete uniform is less skewed than probability vectors simulated using the Dirichlet.}\label{fig-simple_multinomial_3_plots}
\end{figure}

\subsection{Mixed effects model}\label{ssec:mixedeffect}

\subsubsection{Netherlands schools data}

To demonstrate the performance of THAMES on a random effects model, we consider the Netherlands (NL) schools dataset of \citet{snijders1999}. For our purposes, the data consist of language test scores of 2,287 eighth-grade pupils from 133 classes (in 131 schools) in the Netherlands. We denote by $y_{ij}\in\mathbb{R}$ the language test score of pupil $i$ in class $j$, where $j\in\{1,\ldots,J\}$ with $J=133$ and $i\in\{1,\ldots,n_j\}$ with $n_j$ the size of class $j$. Let $n=\sum_{j=1}^J n_j=2,287$ denote the full sample size. 

We aim to determine if there is clustering of language test scores by class, with some classes performing significantly better than others on average. To do this, we fit both a simple mean model (which treats test scores of students in the same class as independent) and a random intercept model (which accounts for correlation of test scores within each class) to the data. The former (null) model $H_0$ posits that all classes perform the same, on average, while the latter (alternative) model $H_1$ allows for variation in performance at the class level. We estimate the log marginal likelihoods for the two models,  $\ell_0(y)$ and $\ell_1(y)$, respectively, using THAMES and bridge sampling \citep{gronau2020} for comparison. With estimates of $\ell_0(y)$ and $\ell_1(y)$, we estimate the log Bayes factor $\text{b}_{01}$ to conduct a Bayesian hypothesis test of $H_0$ versus $H_1$. Note that posterior simulation and marginal likelihood calculation are not analytically tractable for this model. As such, the use of approximate posterior sampling (e.g., via MCMC) and marginal likelihood estimation (e.g., via THAMES) is required.

\subsubsection{Linear model (LM)}

We first consider a simple mean model (denoted LM), which posits that
\begin{align*}
y_{ij} &= \mu + \varepsilon_{ij}, \quad j\in\{1,\ldots,J\}, i\in\{1,\ldots,n_j\}, \sum_j n_j=n, \\
\varepsilon_{ij} & \stackrel{\rm iid}{\sim}  N(0,\sigma_{\varepsilon}^2),  \\ \quad
\mu &\sim N(\hat\mu,\hat\sigma_{\mu}^2), \\
\sigma_{\varepsilon}^2 &\sim\text{InverseGamma}(\hat\nu_{\varepsilon},\hat\beta_{\varepsilon}).
\end{align*}
The fixed hyperparameters $\hat\mu,\hat\sigma_{\mu}^2,\hat\nu_{\varepsilon},\hat\beta_{\varepsilon}$ are specified so as to ensure that the prior distribution is dispersed relative to the likelihood, but on the same scale, as 
\begin{align*}
\hat\mu &= \text{mean}(y_{ij}) = 40.93, \\
\hat\sigma_{\varepsilon}^2 &= \sqrt{2}\cdot\text{sd}(y_{ij}) = 12.73,\\
\hat\nu_{\varepsilon} &= 0.5, \\
\hat\beta_{\varepsilon} &= 0.5\cdot\text{var}(y_{ij}) = 40.53.
\end{align*}
The hyperparameters $(\hat\nu_{\varepsilon},\hat\beta_{\varepsilon})$ are chosen so that the prior mean of the precision $1/\sigma_{\varepsilon}^2$ equals $1/\text{var}(y_{ij})$.
The set of parameters to be estimated in this model $(\mu,\sigma_{\varepsilon}^2)$ has dimension $d=2$.



\subsubsection{Full linear mixed model (full LMM)}

We consider the random intercept model (denoted full LMM):
\begin{align*}
y_{ij} &= \mu + \alpha_j+\varepsilon_{ij}, \\
\varepsilon_{ij} & \stackrel{\rm iid}{\sim} N(0,\sigma_{\varepsilon}^2), \\
\alpha_j & \stackrel{\rm iid}{\sim} N(0,\sigma_{\alpha}^2), \\
\mu &\sim N(\hat\mu,\hat\sigma_{\mu}^2), \\
\sigma_{\varepsilon}^2 &\sim\text{InverseGamma}(\hat\nu_{\varepsilon},\hat\beta_{\varepsilon}), \\
\sigma_{\alpha}^2 &\sim\text{InverseGamma}(\hat\nu_{\alpha},\hat\beta_{\alpha}).
\end{align*}
Here $(\hat\mu,\hat\sigma_{\mu}^2,\hat\nu_{\varepsilon},\hat\beta_{\varepsilon})$ are as above and we specify
\begin{align*}
\hat\nu_{\alpha} &= 0.5, \\
\hat\beta_{\alpha} &= 0.5\cdot\text{var}(\hat\mu_j) = 13.77,
\end{align*}
where $\hat\mu_j = \frac{1}{n_j}\sum_{i=1}^{n_j}y_{ij}$ is the sample mean for class $j\in\{1,\ldots,J\}$. The hyperparameters $(\hat\nu_{\alpha},\hat\beta_{\alpha})$ are chosen so that the prior mean of the precision $1/\sigma_{\alpha}^2$ equals $1/\text{var}(\hat\mu_j)$.
The set of parameters to be estimated in this model $(\mu,\sigma_{\varepsilon}^2,\sigma_{\alpha}^2,\alpha)$ has dimension $d=136$.

\subsubsection{Reduced linear mixed model (reduced LMM)}

Note that the intercept parameters of the full LMM are not identifiable, as there is give-and-take in estimating the grand mean $\mu$ and the random intercepts $\alpha_j$. By absorbing $\alpha_j$ into the error term structure $\varepsilon_{ij}$, we can specify an equivalent model (having the same marginal likelihood) with $d=3$ identifiable parameters $(\mu,\sigma_{\varepsilon}^2,\sigma_{\alpha}^2)$. Mathematically, this amounts to 
marginalizing 
$\alpha_j$ out of the model.
The model (denoted reduced LMM) is given by
\begin{align*}
y_{ij} &= \mu+\varepsilon_{ij}, \\
\varepsilon_{ij} &\sim N(0,\sigma_{\varepsilon}^2+\sigma_{\alpha}^2), \\
\text{Cov}(\varepsilon_{ij},\varepsilon_{i'j}) &= \sigma_\alpha^2, \quad i,i'\in\{1,\ldots,n_j\}, i\neq i', \\
\text{Cov}(\varepsilon_{ij},\varepsilon_{i'j'}) &= 0, \quad j\neq j', \\
\mu &\sim N(\hat\mu,\hat\sigma_{\mu}^2), \\
\sigma_{\varepsilon}^2 &\sim\text{InverseGamma}(\hat\nu_{\varepsilon},\hat\beta_{\varepsilon}), \\
\sigma_{\alpha}^2 &\sim\text{InverseGamma}(\hat\nu_{\alpha},\hat\beta_{\alpha}).
\end{align*}
Here $(\hat\mu,\hat\sigma_{\mu}^2,\hat\nu_{\varepsilon},\hat\beta_{\varepsilon},\hat\nu_{\alpha},\hat\beta_{\alpha})$ are as above. 

\subsubsection{Results}

Figure \ref{nlschools} 
shows the log marginal likelihood of the NL schools data for each model as estimated by THAMES 
with approximate 95\% confidence intervals, as a function of the number of posterior MCMC draws. In the left panel, the bridge sampling estimate of the log marginal likelihood for the reduced LMM based on 20,000 posterior samples is plotted in black \citep{gronau2020}. In the right panel, the bridge sampling estimate for the LM with 20,000 posterior samples is shown. Bridge sampling is a popular state-of-the-art method to estimate log marginal likelihoods from posterior MCMC samples, which is substantially more complicated computationally than the THAMES. MCMC sampling is carried out in R using Stan \citep{r2023,stan2022}. We use values of the posterior sample size $T$ evenly spaced between 2,000 and 20,000. For each $T$, we run 4 chains in parallel for $T/2$ iterations and burn the first $T/4$, yielding $T/4$ MCMC samples from each of the 4 chains.

THAMES provides unbiased estimates of the log marginal likelihood with greater precision as the posterior sample size grows. 
As we would expect, THAMES converges much faster for the LM (with $d=2$) and the reduced LMM (with $d=3$) as compared to the full LMM (with $d=136$), although the estimates of the reduced and full LMM converge to the same value. 
While the posterior support of this model is constrained due to the variance parameters $(\sigma^2_\varepsilon,\sigma^2_\alpha)$, we found that the truncation correction defined in Section \ref{ssec: bounded_param_problem} had no impact on the results.
For a given posterior sample size, we find that bridge sampling generally produces more precise estimates than THAMES. However, THAMES has the advantage of being much simpler to implement in practice.

Using the THAMES estimates of the log marginal likelihoods for the LM ($\ell_0(y)$) and the (reduced) LMM ($\ell_1(y)$) with 20,000 posterior draws, the log Bayes factor ($\text{b}_{01}$) is estimated as
\[
\text{b}_{01} = \ell_0(y)-\ell_1(y) = -8278.842 + 8136.561 = -142.281,
\]
which indicates decisive evidence in favor of the random intercept model \citep{kass1995bayes}.

\begin{figure}
    \centering
    \includegraphics[width=\textwidth]
    {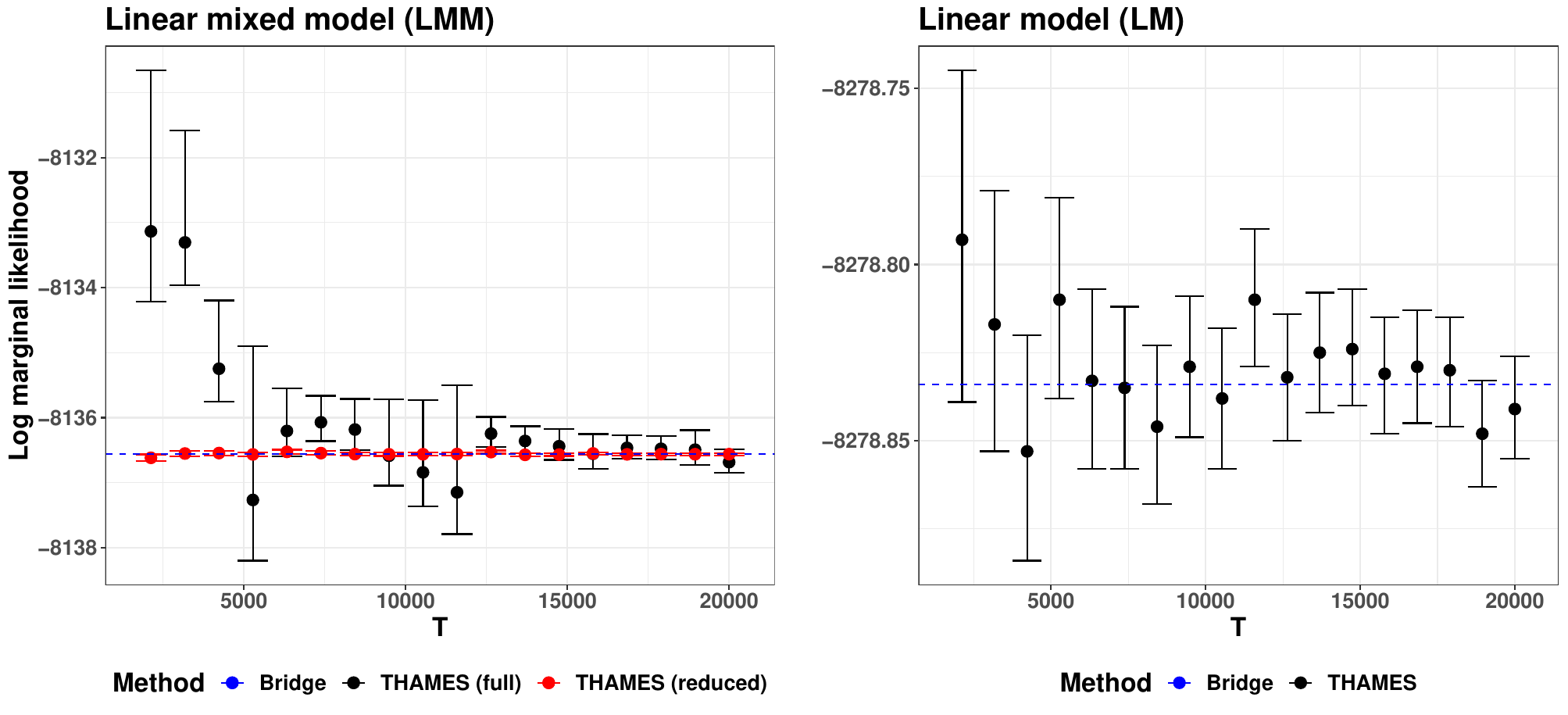}
    \caption{THAMES log marginal likelihood estimates for linear and mixed models fitted to the NL schools data.}
    \label{nlschools}
\end{figure}

\section{Discussion}
\label{sect-discussion}
We have proposed an estimator of the reciprocal of the marginal likelihood, 
called THAMES,
which is simple to compute, unbiased, consistent, has finite variance
and is asymptotically normal, with available confidence intervals.
It is a version of reciprocal importance sampling.
The estimator has one user-specified control parameter, and we
have derived an optimal value for this in the situation where the 
posterior distribution is normal, which is of great interest because
posterior distributions are asymptotically normal in many situations.
We have carried out several numerical experiments in which the estimator
performs well.

The THAMES relies on estimating the posterior covariance matrix and mean. In our experience it is important that the estimator chosen for the covariance matrix be accurate for estimating each matrix entry. Elementwise accuracy appears to be important because the covariance matrix is used to precisely define a quadratic inequality. For example, using a Shrinkage estimator for the covariance matrix, which can produce large errors in a small proportion of its elements, has in our experience degraded the performance of the THAMES in some situations.

One possible alternative to covariance matrix estimation would be to select a minimum-volume covering ellipse which includes a certain percentage of those points of the posterior sample which have the largest value with respect to the (unnormalized) posterior density evaluated at those points. This would ensure that an HPD-region is well approximated, independent of the underlying posterior distribution. Determining a minimum-volume covering ellipse given a set of points can be difficult computationally, but this problem has been  addressed by a vast amount of literature in many different settings and could possibly be adapted to the THAMES. 

\paragraph{Acknowledgements:} 
Irons's research was supported by a Shanahan Endowment Fellowship and a Eunice Kennedy Shriver National Institute of Child Health and Human Development training grant, T32 HD101442-01, to the Center for Studies in Demography \& Ecology at the University of Washington.
Raftery's research was supported by NIH grant
R01 HD070936 from the Eunice Kennedy Shriver National
Institute of Child Health and Human Development (NICHD), 
by the Fondation des Sciences Math\'{e}matiques de Paris (FSMP),
and by Universit\'{e} Paris-Cit\'{e}. 

\bibliographystyle{chicago}
\bibliography{thames}

\hphantom{\cite{NIST}} 

\appendix
\section*{Appendix 1: Proofs of Theorems \ref{thm-limit_behav_c_d}, \ref{thm-limit_behav_SCV} and \ref{thm-n_limit_approx}}

The first two theorems will be proven using two lemmas and one proposition. Theorem \ref{thm-n_limit_approx} is proven using Theorem \ref{thm-limit_behav_c_d}, while Theorem \ref{thm-limit_behav_c_d} is proven using Theorem \ref{thm-limit_behav_SCV} and Lemma \ref{lem-anal_sol_f}. Theorem \ref{thm-limit_behav_SCV} is proven using Proposition \ref{prop-exact_bounds_SCV}, while Proposition \ref{prop-exact_bounds_SCV} is proven using Lemma \ref{lem-anal_sol_f} and Lemma \ref{lem-simple_sol_f}.

Lemma \ref{lem-anal_sol_f} gives the first basic results and an explicit expression of the $SCV$. Lemma \ref{lem-simple_sol_f} can be used to simplify this expression and give an exact upper and an exact lower bound on the $SCV$. This is shown in Proposition \ref{prop-exact_bounds_SCV}. These exact bounds are simplified to their limiting solutions to prove Theorem \ref{thm-limit_behav_SCV}. It is further shown that these bounds lead to a contradiction if the limit of $c_d/d$ is not 1, proving Theorem \ref{thm-limit_behav_c_d}. Finally, the fact that $c_d$ is unique (Theorem \ref{thm-limit_behav_c_d}) is used to prove Theorem \ref{thm-n_limit_approx}.

\begin{lemma}\label{lem-anal_sol_f} Let $m$ denote the mean and $\Sigma$ the covariance matrix of the posterior distribution, which is multivariate normal. There exists a sphere $S$ centered in $0$ with radius $c$ such that $A_{or}=\sqrt{\Sigma}\cdot S+m$ minimizes the variance of the THAMES, i.e.\ any other measurable set $C$ results in a larger or equal variance. Choosing the set $A_{or}$ for the THAMES gives the corresponding $SCV$ \begin{align}
    SCV(d,c)=\underbrace{\frac{d\cdot 2^{d/2}}{1/\Gamma\left({\frac {d}{2}}+1\right)} }_{=:\kappa_d}\cdot\frac{1}{c^{d+2}}\cdot\underbrace{\frac{1}{c^{d-2}} \int_0^c \exp\left(\frac{r^2}{2}\right)r^{d-1}dr}_{=:f(d,c)}-1.
\end{align} It has a unique minimal point $c_d\geq \sqrt{d}$ and its first derivative with respect to $c$ is \begin{align}
    \frac{\partial SCV(d,c)}{\partial c}=\frac{-2d\kappa_d}{c^{d+1}}\frac{1}{c^2}f(d,c)+\frac{\kappa_d}{c^{d+1}}\exp\left(\frac{c^2}{2}\right),
\end{align} where an explicit expression of $f$ is given by \begin{align}\label{equ-analyt_sol_f}
    f(d,c)=(d-2)!!\left(-\frac{1}{c^2}\right)^{\left\lceil\frac{d-2}{2}\right\rceil}\left(f_{d,c}+\sum^{\left\lceil\frac{d-2}{2}\right\rceil}_{r=1}\frac{\exp\left(\frac{c^2}{2}\right)\left(-\frac{1}{c^2}\right)^{-r}}{\left(2r+2\left\lfloor \frac{d}{2}\right\rfloor-d\right)!!}\right),
\end{align} with the double factorial denoting \begin{align*}
    k!!:=\begin{cases}k\cdot(k-2)\cdots 2&k\text{ even,}\\k\cdot(k-2)\cdots3&k\text{ odd,}\\0&k\leq0,\end{cases}
\end{align*} for any $k\in\mathbbm{Z}$ and \begin{align*}f_{d,c}:=\begin{cases}
    \frac{1}{c}\sqrt{\frac{\pi}{2}}\textup{erfi}\left(\frac{c}{\sqrt{2}}\right)& d\text{ odd,}\\ \exp(\frac{c^2}{2})-1& d\text{ even.}\end{cases}\end{align*}\end{lemma}

Note that $SCV$ does not depend on the posterior mean $m$ or the posterior covariance matrix $\Sigma$. Thus any optimal choice of $c$ is going to be optimal for any normal posterior distribution with positive definite covariance matrix.

\begin{proof}
    Let us calculate the following quantity in order to minimize it according to the shape of $A_{or}$:\begin{align*}
    SCV_d(A_{or}):&=\frac{Var_{\theta^{(1)}}\left(\left.\frac{\mathbbm{1}_{A_{or}}(\theta^{(1)})/V(A)}{Zp(\theta^{(1)}|\mathcal{D})}\right|\mathcal{D}\right)}{E_{\theta^{(1)}}\left(\left.\frac{\mathbbm{1}_{A_{or}}(\theta^{(1)})/Z}{p(\mathcal{D})p(\theta^{(1)}|\mathcal{D})}\right|\mathcal{D}\right)^2}\\&=Z^2\left(\int_{A_{or}}\left(\frac{1/V(A_{or})}{Zp(\theta|\mathcal{D})}\right)^2p(\theta|\mathcal{D})\;d\theta-\frac{1}{Z^2}\right)=\int_{A_{or}} \frac{1/V(A_{or})^2}{p(\theta|\mathcal{D})}\;d\theta - 1
\end{align*}

Since $p(\theta|\mathcal{D})=:p(\mu|\mathcal{D})$ is the density of a normal distribution with mean $m$ an positive definite covariance matrix $\Sigma$ it follows that \begin{align*}
SCV_d(A_{or})&=\int_{A_{or}}\frac{1/V(A_{or})^2}{p(\mu|\mathcal{D})}d\mu-1\\
&=\frac{\sqrt{|\Sigma|}(2\pi)^{d/2}}{V(A_{or})^2}\int_{A_{or}} \exp\left(\frac{1}{2}(\mu-m)^t\Sigma^{-1}(\mu-m)\right)d\mu-1\\
&=\frac{\sqrt{|\Sigma|}(2\pi)^{d/2}}{V(A_{or})^2}\int_{T(A_{or},m,\Sigma)} \exp\left(\frac{\mu^t\mu}{2}\right)d\mu-1,
 \end{align*}
where $T(A_{or},m,\Sigma)=(\sqrt{\Sigma})^{-1}(A_{or}-m)=S$ denotes the translate of $A_{or}$ by $-m$ and the rescaling of $A_{or}$ by $(\sqrt{\Sigma})^{-1}$. From this it immediately follows that there exists a $c$ such that $A_{or}$ is optimal. 

\paragraph{Proving that $A_{or}$ is optimal}

Let $C$ be a measurable set with volume $V(C)\in (0,\infty)$. We can choose $c$ such that \begin{align*}
    V(S)=V(T(C,m,\Sigma))\Leftrightarrow V(A_{or})/\sqrt{|\Sigma|}=V(C)/\sqrt{|\Sigma|}.
\end{align*} This directly implies $V(A_{or})=V(C)$. It follows that \begin{align*}
    SCV_d(C)&=\frac{|\Sigma|(2\pi)^{d/2}}{V(C)^2}\int_{T(C,m,\Sigma)} \exp\left(\frac{\mu^t\mu}{2}\right)d\mu-1\\&=\frac{|\Sigma|(2\pi)^{d/2}}{V(A_{or})^2}\int_{T(C,m,\Sigma)\cap S} \exp\frac{\mu^t\mu}{2}d\mu\\&+\frac{|\Sigma|(2\pi)^{d/2}}{V(A_{or})^2}\int_{T(C,m,\Sigma)/S} \exp\left(\frac{\mu^t\mu}{2}\right)d\mu-1\\&\geq \frac{|\Sigma|(2\pi)^{d/2}}{V(A_{or})^2}\left(\int_{T(C,m,\Sigma)\cap S} \exp\frac{\mu^t\mu}{2}d\mu+\exp\left(\frac{c^2}{2}\right)V(T(C,m,\Sigma)/S)\right)-1.
\end{align*} This is due to the fact that $T(C,m,\Sigma)/S=\{\mu\in T(C,m,\Sigma)|\mu^T\mu\geq c^2\}.$ Using this and the fact that $V(S/T(C,m,\Sigma)=V(T(C,m,\Sigma)/S)$ one can similarly conclude \begin{align*}
    SCV_d(C)&\geq\frac{|\Sigma|(2\pi)^{d/2}}{V(A_{or})^2}\int_{T(C,m,\Sigma)\cap S} \exp\frac{\mu^t\mu}{2}d\mu\\&+\frac{|\Sigma|(2\pi)^{d/2}}{V(A_{or})^2}\int_{S/T(C,m,\Sigma)} \exp\left(\frac{\mu^t\mu}{2}\right)d\mu-1=SCV_d(A_{or}).
\end{align*} 
It follows that for any measurable set $C$ a $c$ exists such that the $SCV$ (and thus the variance) of the THAMES is strictly lower when choosing $h$ to be the uniform distribution on $A_{or}$ instead of $C$.

\paragraph{Calculating an exact expression for the SCV}

 Let us remind that the area of a $d$-sphere of radius $r$ is $2\pi^{\frac{d}{2}}r^{d-1}/\Gamma\left({\frac {d}{2}}\right)$ and its volume is $\pi ^{d/2}r^{d}/\Gamma\left(\frac{d}{2}+1\right)$. Using the spherical coordinates, we get that $SCV_d(A_{or})=SCV(d,c)$ defined as follows for $d\geq 1$:\begin{align*}
SCV(d,c)
&=\frac{|\Sigma|}{V(A_{or})^2}\int_{T(A_{or},m,\Sigma)} (2\pi)^{d/2}\exp\left(\frac{\mu^t\mu}{2}\right)d\mu-1\\
&=\frac{(2\pi)^{d/2}|\Sigma|}{V(A_{or})^2}\int_0^c \exp\left(\frac{r^2}{2}\right)2\pi^{\frac{d}{2}}r^{d-1}/\Gamma\left({\frac {d}{2}}\right)dr-1\\
&=\frac{(2\pi)^{d/2}|\Sigma|2\pi^{\frac{d}{2}}}{\Gamma\left({\frac {d}{2}}\right)} \times \frac{1}{V(A_{or})^2}\int_0^c \exp\left(\frac{r^2}{2}\right)r^{d-1}dr-1\\
&=\frac{(2\pi)^{d/2}|\Sigma|2\pi^{\frac{d}{2}}}{\Gamma\left({\frac {d}{2}}\right)} \times \frac{1}{\pi ^{d}c^{2d}|\Sigma|/\Gamma\left(\frac{d}{2}+1\right)^2}\int_0^c \exp\left(\frac{r^2}{2}\right)r^{d-1}dr-1\\
&=\frac{d\cdot 2^{d/2}}{c^{2d}/\Gamma\left({\frac {d}{2}}+1\right)} \int_0^c \exp\left(\frac{r^2}{2}\right)r^{d-1}dr-1\\
&=\frac{\kappa_d}{c^{2d}}\int_0^c \exp\left(\frac{r^2}{2}\right)r^{d-1}dr-1
\end{align*}

First, the limiting behaviour of $SCV(d,c)$ w.r.t.\ $c$ can be observed by applying l'Hôpital's rule. We have \begin{align*}
    \lim_{c\to\infty}SCV(d,c)\propto\lim_{c\to\infty}\frac{\int^c_0\exp\left(\frac{r^2}{2}\right)r^{d-1}\;dr}{c^{2d}}=\lim_{c\to\infty}\frac{\exp\left(\frac{c^2}{2}\right)c^{d-1}}{(2d)\cdot c^{2d-1}}=\infty
\end{align*} and\begin{align*}
    \lim_{c\downarrow0}SCV(d,c)\propto\lim_{c\downarrow0}\frac{\int^c_0\exp\left(\frac{r^2}{2}\right)r^{d-1}\;dr}{c^{2d}}=\lim_{c\downarrow0}\frac{\exp\left(\frac{c^2}{2}\right)c^{d-1}}{(2d)\cdot c^{2d-1}}=\infty.
\end{align*}

Since $SCV(d,c)$ is continuous w.r.t.\ $c$ it follows that a global minimum exists and that this minimum can be found by setting the first derivative of $SCV(d,c)$ w.r.t.\ $c$ to 0. 

Let us now take the derivative of $SCV(d,c)$ w.r.t.\ $c$:
\begin{align*}
\frac{\partial SCV(d,c)}{\partial c}&=
\frac{-2d\kappa_d}{c^{2d+1}}\int_0^c \exp\left(\frac{r^2}{2}\right)r^{d-1}dr + \frac{\kappa_d}{c^{2d}}\exp\left(\frac{c^2}{2}\right)c^{d-1} \\
& = \frac{-2d\kappa_d}{c^{d+1}}\frac{1}{c^2}f(d,c)+\frac{\kappa_d}{c^{d+1}}\exp\left(\frac{c^2}{2}\right)
\end{align*}

The first order condition is thus \begin{align}
\frac{-2d\kappa_d}{c^{d+1}}\frac{1}{c^2}f(d,c)+\frac{\kappa_d}{c^{d+1}}\exp\left(\frac{c^2}{2}\right)=0\Leftrightarrow\frac{\frac{d}{c^2}f(d,c)}{\exp\left(\frac{c^2}{2}\right)} =\frac{1}{2},\label{equ-1st_order_condition_Q_c}
\end{align} or equivalently \begin{align}
    \frac{2d}{c^d}\int^c_0\exp\left(\frac{r^2}{2}\right)r^{d-1}\;dr=\exp\left(\frac{c^2}{2}\right).\label{equ-1st_order_condition_without_f}
\end{align}

Let $c_d$ be a point that fulfills Equation \eqref{equ-1st_order_condition_without_f}. We plug $c_d$ into the second order condition:\begin{align*}
    \left.\frac{\partial^2 SCV(d,c)}{(\partial c)^2}\right|_{c=c_d}&=\frac{-2d(-2d-1)\kappa_d}{c_d^{2d+2}}\int_0^{c_d} \exp\left(\frac{r^2}{2}\right)r^{d-1}dr \\&+ \frac{-2d\kappa_d}{c_d^{2d+1}}\exp\left(\frac{c_d^2}{2}\right)c_d^{d-1}\\&+\frac{\kappa_d(-d-1)}{c_d^{d+2}}\exp\left(\frac{c_d^2}{2}\right)+\frac{\kappa_d}{c_d^{d}}\exp\left(\frac{c_d^2}{2}\right)\\&=
     \kappa_d\frac{\exp\left(\frac{c_d^2}{2}\right)}{c_d^d}\left(1-\frac{d}{c_d^2}\right) \geq 0
\end{align*} This is the case if, and only if $c_d\geq \sqrt{d}$. Further any local maximum point $c$ needs to fulfill $c\leq \sqrt{d}$. Since any two local minima need to have at least 1 local maximum in between them, it follows that there exists only one minimum point, $c_d$. Still, the explicit solution of $f$ needs to be verified.

For $d\geq 1$,
\begin{align}
f(d,c) &= \frac{1}{c^{d-2}}\int_0^c \exp\left(\frac{r^2}{2}\right)r^{d-1}dr \\
&=\frac{1}{c^{d-2}}\left\{\left[\exp\left(\frac{r^2}{2}\right)r^{d-2}\right]_0^c-\int_0^c (d-2) \exp\left(\frac{r^2}{2}\right)r^{d-3}dr\right\}\\
&=\frac{1}{c^{d-2}}\left\{\exp\left(\frac{c^2}{2}\right)c^{d-2}-\mathbbm{1}\left(d=2\right)\right\}-\frac{d-2}{c^{d-2}}\int_0^c  \exp\left(\frac{r^2}{2}\right)r^{d-3}dr\\
&= :     \exp\left(\frac{1}{2}c^2\right)-\mathbbm{1}\left(d=2\right)-(d-2)\frac{1}{c^2}f(d-2,c),\label{equ-rec_relation_f}
\end{align}
with $f(0,c)=0$ and $f(1,c)=c\sqrt{\frac{\pi}{2}}\text{erfi}\left(\frac{c}{\sqrt{2}}\right)$.

To verify the explicit solution of $f$, it suffices to show that it fulfills the recursive relation \eqref{equ-rec_relation_f} for $d\geq 3$ and that the initial value conditions for $d=1,d=2$ hold.

\;

This is indeed the case, as this expression for $f$ fulfills the initial value conditions by definition and {\footnotesize
\begin{align*}
    &-f(d-2,c)\frac{1}{c^2}(d-2)+\exp\left(\frac{c^2}{2}\right)\\&=\exp\left(\frac{c^2}{2}\right)+(d-2)!!\left(-\frac{1}{c^2}\right)^{\left\lceil\frac{d-2}{2}\right\rceil}\left(f_{d-2,c}+\sum^{\left\lceil\frac{d-2}{2}\right\rceil-1}_{r=1}\frac{\exp\left(\frac{c^2}{2}\right)\left(-\frac{1}{c^2}\right)^{-r}}{\left(2r+2\left\lfloor \frac{d}{2}\right\rfloor-d\right)!!}\right)\\&=\exp\left(\frac{c^2}{2}\right)+(d-2)!!\left(-\frac{1}{c^2}\right)^{\left\lceil\frac{d-2}{2}\right\rceil}\cdot\left(f_{d,c}-\frac{\exp\left(\frac{c^2}{2}\right)}{(d-2)!!}\left(-\frac{1}{c^2}\right)^{-\left\lceil\frac{d-2}{2}\right\rceil}+\sum^{\left\lceil\frac{d-2}{2}\right\rceil}_{r=1}\frac{\exp\left(\frac{c^2}{2}\right)\left(-\frac{1}{c^2}\right)^{-r}}{\left(2r+2\left\lfloor \frac{d}{2}\right\rfloor-d\right)!!}\right)\\&=(d-2)!!\left(-\frac{1}{c^2}\right)^{\left\lceil\frac{d-2}{2}\right\rceil}\left(f_{d,c}+\sum^{\left\lceil\frac{d-2}{2}\right\rceil}_{r=1}\frac{\exp\left(\frac{c^2}{2}\right)\left(-\frac{1}{c^2}\right)^{-r}}{\left(2r+2\left\lfloor \frac{d}{2}\right\rfloor-d\right)!!}\right)\\&=f(d,c).
\end{align*}
}%
\end{proof}

\begin{lemma}\label{lem-simple_sol_f}
The analytical solution of $f$ (Equation \eqref{equ-analyt_sol_f}) can be simplified to \begin{align}f(d,c)=-\exp\left(\frac{c^2}{2}\right)\sum_{r=\left\lceil\frac{d-2}{2}\right\rceil+1}^{\infty}\frac{(d-2)!!\left(-\frac{1}{c^2}\right)^{\left\lceil\frac{d-2}{2}\right\rceil-r}}{\left(2r+2\left\lfloor \frac{d}{2}\right\rfloor-d\right)!!}.
\end{align} 
\end{lemma}

Lemma \ref{lem-simple_sol_f} provides a closed form expression of $f$ for both the even and uneven case, but is likely only of theoretical interest, as it is not clear how this sum should be implemented numerically.

\begin{proof}
    We are trying to show that \begin{align*}&f(d,c_d)=
    (d-2)!!\left(-\frac{1}{c_d^2}\right)^{\left\lceil\frac{d-2}{2}\right\rceil}\left(f_{d,c}+\sum^{\left\lceil\frac{d-2}{2}\right\rceil}_{r=1}\frac{\exp\left(\frac{c_d^2}{2}\right)\left(-\frac{1}{c_d^2}\right)^{-r}}{\left(2r+2\left\lfloor \frac{d}{2}\right\rfloor-d\right)!!}\right)\\& =-\exp\left(\frac{c_d^2}{2}\right)\sum_{r=\left\lceil\frac{d-2}{2}\right\rceil+1}^{\infty}(d-2)!!\left.\left(-\frac{1}{c_d^2}\right)^{\left\lceil\frac{d-2}{2}\right\rceil-r}\right/\left(2r+2\left\lfloor \frac{d}{2}\right\rfloor-d\right)!!.
\end{align*} 

This simplification is a consequence of the fact that the sum is a Taylor series of $-f_{d,c}\exp\left(\frac{c_d^2}{2}\right)$ in both the even and the uneven case.
\newpage
\textbf{Case 1: $d$ is even}

We make use of the following identity: \begin{align*}
    (2r)!!=2\cdot 4\cdots 2r = 2^r (1\cdot 2\cdots r)=2^rr!
\end{align*}

Thus for $d$ even \begin{align*}
    \sum^{\infty}_{r=1}\frac{\exp\left(\frac{c_d^2}{2}\right)\left(-\frac{1}{c_d^2}\right)^{-r}}{\left(2r+2\left\lfloor \frac{d}{2}\right\rfloor-d\right)!!}&=\sum^{\infty}_{r=1}\frac{\exp\left(\frac{c_d^2}{2}\right)\left(-\frac{1}{c_d^2}\right)^{-r}}{(2r)!!}\\=\sum^{\infty}_{r=1}\frac{\exp\left(\frac{c_d^2}{2}\right)\left(-\frac{c_d^2}{2}\right)^{r}}{r!}&=-\left(\exp\left(\frac{c_d^2}{2}\right)-1\right)=-f_{d,c}.
\end{align*}

\textbf{Case 2: $d$ is uneven}

We make use of the following identity:\begin{align*}
    (2r-1)!!=\sqrt{\frac{2}{\pi}}2^{r-1/2}\Gamma(r+1/2)
\end{align*}

Thus for $d$ odd \begin{align*}
    &\sum^{\infty}_{r=1}\frac{\exp\left(\frac{c_d^2}{2}\right)\left(-\frac{1}{c_d^2}\right)^{-r}}{\left(2r+2\left\lfloor \frac{d}{2}\right\rfloor-d\right)!!}=\sum^{\infty}_{r=1}\frac{\exp\left(\frac{c_d^2}{2}\right)\left(-\frac{1}{c_d^2}\right)^{-r}}{(2r-1)!!}\\&=\sqrt{\pi}\sum^{\infty}_{r=1}\frac{\exp\left(\frac{c_d^2}{2}\right)\left(-\frac{c_d^2}{2}\right)^{r}}{\Gamma(r+1/2)}=-\frac{\sqrt{\pi}}{2}c_d^2\sum^\infty_{r=0}\frac{\exp\left(\frac{c_d^2}{2}\right)\left(-\frac{c_d^2}{2}\right)^{r}}{\Gamma(r+3/2)}\\&=-\frac{\sqrt{\pi}}{2}c_d^2\left(\sum^\infty_{r=0}\left(\frac{c_d^2}{2}\right)^r\frac{1}{r!}\right)\left(\sum^\infty_{r=0}\frac{\left(\frac{c_d^2}{2}\right)^{r}(-1)^r}{\Gamma(r+3/2)}\right)\\&=-\frac{\sqrt{\pi}}{2}c_d^2\sum^\infty_{r=0}\left(\frac{c_d^2}{2}\right)^{r}\sum^r_{k=0}\frac{1}{(r-k)!}\cdot\frac{(-1)^r}{\Gamma(r+3/2)}\\&=-\frac{\sqrt{\pi}}{2}c_d^2\sum^\infty_{r=0}\left(\frac{c_d^2}{2}\right)^{r}\sum^r_{k=0}\frac{(-1)^r/\Gamma(r+3/2)}{\Gamma(r-k+1)}=-\frac{\sqrt{\pi}}{2}c_d^2\sum^\infty_{r=0}\left(\frac{c_d^2}{2}\right)^{r}\frac{2/\Gamma(r+1)}{\sqrt{\pi}(2r+1)}\\&=-\sqrt{2}c_d\sum^\infty_{r=0}\left(\frac{c_d}{\sqrt{2}}\right)^{2r+1}\frac{1}{(2r+1)r!}=-\sqrt{\frac{\pi}{2}}c_d\text{erfi}\left(\frac{c_d}{\sqrt{2}}\right)=-f_{d,c},
\end{align*}

where we used the Cauchy product rule and a specific solution for the gaussian hypergeometric function evaluated at $(1,-r,\frac{3}{2},1)$. By (Olver et al.,§15.4(i), Equation 15.2.4) \begin{align*}
    F_{hypergeometric}(-r,1,\frac{3}{2},1)&=\sum^r_{k=0}\frac{\Gamma(1+k)}{\Gamma(1)}\cdot(-1)^k\binom{r}{k}\cdot\frac{\Gamma(3/2)}{\Gamma(k+3/2)}\\&=\sum^r_{k=0}k!\cdot\frac{(-1)^kr!}{(r-k)!k!}\cdot \frac{\sqrt{\pi}/2}{\Gamma(k+3/2)}\\&=\Gamma(r+1)\sqrt{\pi}/2\sum^r_{k=0}\frac{(-1)^k}{\Gamma(r-k+1)\Gamma(k+3/2)}.
\end{align*} Further by (Olver et al., §15.4(i), Equation 15.4.20)\begin{align*}
    F_{hypergeometric}(-r,1,\frac{3}{2},1)&=\frac{\Gamma(3/2)\Gamma(3/2-1+r)}{\Gamma(3/2-1)\Gamma(3/2+r)}\\&=\frac{1}{2}\cdot\frac{1}{(1/2+r)}=\frac{1}{(2r+1)}.
\end{align*}


\end{proof}

It turns out that we can use this lemma to construct exact upper and lower bounds on the $SCV$.

\begin{proposition}\label{prop-exact_bounds_SCV} For any $c\in(0,\infty)$ which can depend on $d$\begin{align}
    SCV(d,c)\geq \Gamma\left({\frac {d}{2}}+1\right)\frac{2^{d/2-1}}{d^{d/2}}\exp\left(\frac{d}{2}\right)-1
\end{align} and for $c= \sqrt{d+1}$ \begin{align}\label{eq-SCV-exact-upperbound}
    SCV(d,c)\leq \Gamma\left({\frac {d}{2}}+1\right)\frac{2^{d/2}}{c^d}\exp\left(\frac{c^2}{2}\right)-1.
\end{align} For $c=\sqrt{d+L},L\in\mathbb{R}$, there exists a function $SCV'(d,c)$ such that \\$\lim_{d\to\infty}\frac{SCV(d,\sqrt{d+L})}{SCV'(d,\sqrt{d+L})}=1$ and inequality \eqref{eq-SCV-exact-upperbound} holds for $SCV'(d,c)$ in $c=\sqrt{d+L}$.\end{proposition}


\begin{proof}
    We know that the minimum point fulfills the first order condition (Equation \eqref{equ-1st_order_condition_Q_c}): \begin{align*}
    \frac{\frac{d}{c_d^2}f(d,c_d)}{\exp\left(\frac{c_d^2}{2}\right)}=\frac{1}{2}.\end{align*} 

This information can immediately be used to get a lower bound on $SCV(d,c)$ in $c=c_d$. Let $c_d^2=:L_d+d$. By Equation \eqref{equ-1st_order_condition_without_f} \begin{align*}
    \frac{2d}{c_d^2}f(d,c_d)=\frac{2d}{c^{d}_d}\int_0^{c_d} \exp\left(\frac{r^2}{2}\right)r^{d-1}dr.
\end{align*} 
It follows that
\begin{align*}
    SCV(d,c_d)&=\frac{d\cdot 2^{d/2}}{c_d^{2d}/\Gamma\left({\frac {d}{2}}+1\right)} \int_0^{c_d} \exp\left(\frac{r^2}{2}\right)r^{d-1}dr-1\\&=\frac{2^{d/2-1}}{1/\Gamma\left({\frac {d}{2}}+1\right)}\cdot\frac{1}{(d+L_d)^{d/2}}\exp\left(\frac{d+L_d}{2}\right)-1\\&\geq\frac{2^{d/2-1}}{1/\Gamma\left({\frac {d}{2}}+1\right)}\cdot\frac{1}{d^{d/2}}\exp\left(\frac{d}{2}\right)-1
\end{align*}

For this we used the fact that the function \begin{align*}
    L\mapsto \frac{\exp\left(\frac{L-2}{2}\right)}{\left(d+L\right)^{\frac{d}{2}}}
\end{align*} is minimized w.r.t.\ $L\in (-d,\infty)$ if, and only if \begin{align*}
    \frac{\partial}{\partial L}\left(\frac{1}{d+L}\right)^{\frac{d}{2}}\exp\left(\frac{L}{2}\right)=0\\\Leftrightarrow \exp\left(-\ln(d+L)\frac{d}{2}+\frac{L}{2}\right)\left(\frac{1}{2}-\frac{d}{2(d+L)}\right)=0\Leftrightarrow L=0.
\end{align*} Note that this is really a global minimum as the derivative changes signs at $L=0$. Since $c_d$ is a minimal point, this gives a real lower bound for any choice of $c$. 

Determining the upper bound is similar, only that we now want to show that in the point $c=\sqrt{d+L},L\in\mathbb{R}$ the ratio which we used is bounded: \begin{align*}
    \frac{\frac{d}{(\sqrt{d+L})^2}f(d,\sqrt{d+L})}{\exp\left(\frac{(\sqrt{d+L})^2}{2}\right)} \leq1.\end{align*}  Due to Lemma \ref{lem-simple_sol_f} this is equivalent to \begin{align*}
    \sum_{r=\left\lceil\frac{d-2}{2}\right\rceil+1}^{\infty}d!!\left.\left(-\frac{1}{d+L}\right)^{\left\lceil\frac{d-2}{2}\right\rceil+1-r}\right/\left(2r+2\left\lfloor \frac{d}{2}\right\rfloor-d\right)!!&\leq1\\\Leftrightarrow \sum_{k=0}^{\infty}\frac{d!!\left(-(d+L)\right)^{k}}{\left(2k+d\right)!!}&\leq1.
\end{align*}

The sum is absolutely convergent, as it is a power series with infinite convergence radius (see the proof of Lemma \ref{lem-simple_sol_f}: Taylor series are power series, and in our case we know that these specific Taylor series converge everywhere). It follows that we are allowed to change the order of summation. It follows that, if $d+L\leq d+4$, \begin{align*}
    \sum_{k=0}^{\infty}\frac{d!!\left(-(d+L)\right)^{k}}{\left(2k+d\right)!!}&=1+\sum_{k=1}^{\infty}\frac{d!!\left(-(d+L)\right)^{k}}{\left(2k+d\right)!!}\\&=1+\sum_{k=1,k\text{ odd}}^{\infty}\left(\frac{d!!\left(-(d+L)\right)^{k}}{\left(2k+d\right)!!}+\frac{d!!\left(-(d+L)\right)^{k+1}}{\left(2k+d+2\right)!!}\right)\\&=1+\sum_{k=1,k\text{ odd}}^{\infty}\frac{d!!\left(-(d+L)\right)^{k}}{\left(2k+d\right)!!}\left(1-\frac{d+L}{d+2k+2}\right)\\&=1-\sum_{r=0}^{\infty}\frac{d!!\left(d+L\right)^{2r+1}}{\left(4r+d+2\right)!!}\left(1-\frac{d+L}{d+4r+4}\right)\leq 1,
\end{align*} as this implies that any component of the sum is non-negative. No restrictions on $d+L,L\in\mathbb{R}$ are necessary in the limiting case, since for the first finitely many negative expressions\begin{align*}
    \lim_{d\to\infty}\sum^{2\lceil |L|\rceil}_{k=0}\frac{d!!(-(d+L))^k}{(2k+d)!!}=\sum^{2\lceil |L|\rceil}_{k=0}(-1)^k=0.
\end{align*} $SCV'$ is then defined as the $SCV$, but with these first finitely many components set to 0.

This gives \begin{align*}
    SCV(d,\sqrt{d+L})&=\frac{d\cdot 2^{d/2}}{1/\Gamma\left({\frac {d}{2}}+1\right)}\cdot\frac{1}{(\sqrt{d+L})^{2d}} \int_0^{\sqrt{d+L}} \exp\left(\frac{r^2}{2}\right)r^{d-1}dr-1\\&\leq\frac{2^{d/2}}{1/\Gamma\left({\frac {d}{2}}+1\right)}\cdot\frac{1}{(\sqrt{d+L})^{d}}\exp\left(\frac{d+L}{2}\right)-1
\end{align*}

if $d+L\leq d+4$, so in particular if $L=1$. On the other hand the inequality holds for $SCV'$ for any $L\in\mathbb{R}$.
\end{proof}

Calculating the asymptotic behaviour of these bounds gives the proof of Theorem \ref{thm-limit_behav_SCV}:

\begin{proof}[Proof of Theorem \ref{thm-limit_behav_SCV}]

We are first going to prove the limiting behaviour of the $SCV$ w.r.t.\ $d$ and then use this result to provide a proof of the limiting behaviour of the optimal value $c_d$.

Applying Stirling's formula gives \begin{align*}
    SCV(d,c_d)&\geq\frac{2^{d/2-1}}{1/\Gamma\left({\frac {d}{2}}+1\right)}\cdot\frac{1}{d^{d/2}}\exp\left(\frac{d}{2}\right)-1\\&\geq\frac{2^{d/2-1}}{d^{d/2}}\left(\sqrt{2\pi}\left(\frac{d}{2}+1\right)^{\frac{d}{2}+1-\frac{1}{2}}\right)\exp\left(\frac{d}{2}-\frac{d}{2}-1\right)-1\\&=\frac{2^{d/2}\sqrt{\pi (d+2)/4}}{d^{d/2}}\left(\left(\frac{d}{2}+1\right)^{\frac{d}{2}}\right)\exp\left(-1\right)-1\\&=\sqrt{\pi (d+2)/4}\left(\left(\frac{d+2}{d}\right)^{\frac{d}{2}}\right)\exp\left(-1\right)-1\simeq\sqrt{\pi (d+2)/4}-1.
\end{align*} 
The derivative of the logarithm of the function \begin{align*}
    d\mapsto \left(\frac{d+2}{d}\right)^{d}=\exp\left(d\ln\left(\frac{d+2}{d}\right)\right)
\end{align*} is \begin{align*}
    \ln\left(\frac{d+2}{d}\right)-\frac{2}{d+2}
\end{align*} and its second derivative is\begin{align*}
    \frac{1}{d+2}-\frac{1}{d}+\frac{2}{(d+2)^2}<0\Leftrightarrow (d+2)d-(d+2)^2+2d=-4<0.
\end{align*} Thus the first derivative is decreasing towards its limit, 0, which implies that the first derivative is strictly greater than 0. It follows that this function is strictly increasing in $d\geq 1$ with supremum $\exp(2)$ and minimum \begin{align*}
    \left.\left(\frac{d+2}{d}\right)^{d}\right|_{d=1}=3.
\end{align*}

An exact inequality is thus given by \begin{align*}
    SCV(d,c_d)&\geq \exp(-1)\sqrt{3}\sqrt{\pi(d+2)/4}-1\\&\geq 0.63\sqrt{\pi(d+2)/4}-1
\end{align*}

Analogously, applying the other direction of Stirling's formula \begin{align*}SCV(d,\sqrt{d+1})&\leq \frac{2^{d/2}}{1/\Gamma\left({\frac {d}{2}}+1\right)}\cdot\frac{1}{(d+1)^{\frac{d}{2}}}\exp\left(\frac{d+1}{2}\right)-1\\&\leq2\sqrt{\pi/4}\left(\frac{d+2}{d+1}\right)^{\frac{d}{2}}\sqrt{d+2}\exp\left(-\frac{1}{2}+\frac{1}{12d}\right)-1\\&\simeq2\cdot \sqrt{(d+2)\pi/4}-1.\end{align*}

with an exact inequality given by \begin{align*}
    SCV(d,\sqrt{d+1})&\leq 2\sqrt{\pi/4}\left(\frac{d+2}{d+1}\right)^{\frac{d}{2}}\sqrt{d+2}\exp\left(-\frac{1}{2}+\frac{1}{12d}\right)-1\\&\leq 2\sqrt{(d+2)\pi/4}\exp(1/12)-1\leq 1.09\cdot 2\sqrt{(d+2)\pi/4}-1.
\end{align*}
One can analogously apply the asymptotic Stirling's formula for $c=\sqrt{d+L},L\in\mathbb{R}$:\begin{align}
    SCV(d,\sqrt{d+L})&\simeq SCV'(d,\sqrt{d+L})\leq \frac{2^{d/2}}{1/\Gamma\left({\frac {d}{2}}+1\right)}\cdot\frac{1}{(d+L)^{\frac{d}{2}}}\exp\left(\frac{d+L}{2}\right)-1\\&\simeq2\sqrt{\pi/4}\left(\frac{d+2}{d+L}\right)^{\frac{d}{2}}\sqrt{d+2}\exp\left(-\frac{L}{2}\right)-1\\&\simeq2\cdot \sqrt{(d+2)\pi/4}-1.
\end{align}

In total\begin{align*}
    1\leq\liminf_{d\to\infty}\frac{SCV(d,\sqrt{d+L_d})}{\sqrt{\pi(d+2)/4}}\leq\limsup_{d\to\infty}\frac{SCV(d,\sqrt{d+L})}{\sqrt{\pi(d+2)/4}}\leq 2,
\end{align*} for all $L\in\mathbb{R}$. This is equivalent to Theorem \ref{thm-limit_behav_SCV}.\end{proof}\begin{proof}[Proof of Theorem \ref{thm-limit_behav_c_d}]

First of all, the fact that $A_{or}$ with radius $c_d=\sqrt{d+L_d}$ minimizes the $SCV$ and $L_d\geq0$ is a result of Lemma \ref{lem-anal_sol_f}. 

To prove that $\frac{L_d}{d}\stackrel{d\to\infty}\to0$, we show that \begin{align*}
    \limsup_{d\to\infty}\frac{L_d}{d}\leq 0
\end{align*} and \begin{align*}
    \liminf_{d\to\infty}\frac{L_d}{d}\geq0.
\end{align*} Note that \begin{align*}
    2\geq \limsup_{d\to\infty}\frac{SCV(d\sqrt{d+L_d})}{\sqrt{\pi(d+2)/4}}\geq \limsup_{d\to\infty} \left(\frac{d+2}{d+L_d}\right)^{\frac{d}{2}}\exp\left(\frac{L_d-2}{2}\right),
\end{align*} so any $L_d$ for which this expression diverges can be excluded. Further the function \begin{align*}
    L\mapsto \left(\frac{d+2}{d+L}\right)^{\frac{d}{2}}\exp\left(\frac{L-2}{2}\right)
\end{align*} is decreasing if $-d<L<0$ and increasing if $L>0$. Both of these facts follow from previous statements.

Suppose that there exists an $\alpha\in(0,1)$ such that \begin{align*}
    \limsup_{d\to\infty}\frac{L_d}{d}> \alpha.
\end{align*} This implies that \begin{align*}
    L_d>d\alpha
\end{align*} for infinitely many $d$. For these $d$ it holds that \begin{align*}
    \left(\frac{d+2}{d+L_d}\right)^{\frac{d}{2}}\exp\left(\frac{L_d-2}{2}\right)\geq \left(\frac{d+2}{d+d\alpha}\right)^{\frac{d}{2}}\exp\left(\frac{d\alpha-2}{2}\right)\\\propto \left(\frac{d+2}{d}\right)^{\frac{d}{2}}\left(\frac{\exp(\alpha)}{1+\alpha}\right)^{\frac{d}{2}}\simeq \exp(1)\left(\frac{\exp(\alpha)}{1+\alpha}\right)^{\frac{d}{2}}.
\end{align*}

This sequence diverges because of the well known inequality $\exp(\alpha)-1>\alpha$ if $\alpha\neq0$. It follows that \begin{align*}
    \limsup_{d\to\infty}\frac{L_d}{d}\leq \alpha\forall \alpha\in(0,1)\Rightarrow \limsup_{d\to\infty}\frac{L_d}{d}\leq0.
\end{align*}

If on the other hand \begin{align*}
    \liminf_{d\to\infty}\frac{L_d}{d}<-\alpha
\end{align*} we have that \begin{align*}
    L_d<-d\alpha
\end{align*} for infinitely many $d$ and by the same argument \begin{align*}
    \left(\frac{d+2}{d+L_d}\right)^{\frac{d}{2}}\exp\left(\frac{L_d-2}{2}\right)\geq \left(\frac{d+2}{d-d\alpha}\right)^{\frac{d}{2}}\exp\left(\frac{-d\alpha-2}{2}\right)\\\propto \left(\frac{d+2}{d}\right)^{\frac{d}{2}}\left(\frac{\exp(-\alpha)}{1-\alpha}\right)^{\frac{d}{2}}\simeq \exp(1)\left(\frac{\exp(-\alpha)}{1-\alpha}\right)^{\frac{d}{2}}
\end{align*} we have a contradiction. So \begin{align*}
    \liminf_{d\to\infty}\frac{L_d}{d}\geq-\alpha\forall\alpha\in(0,1)\Rightarrow \liminf_{d\to\infty}\frac{L_d}{d}\geq0.
\end{align*} In total \begin{align*}
    0\leq \liminf_{d\to\infty}\frac{L_d}{d}\leq \limsup_{d\to\infty}\frac{L_d}{d}\leq 0.
\end{align*}\end{proof}\begin{proof}[Proof of Theorem \ref{thm-n_limit_approx}]

Let \begin{align*}
    A_n:=\{\theta|\|(\theta-m_n)^T\Sigma_n^{-1}(\theta-m_n)\|^2<c^2\}
\end{align*} be the set used for the THAMES when applying it to estimate the marginal $Z_n^{-1}$ corresponding to $p_n(\theta|\mathcal{D}_n)$. The corresponding $SCV$ is then defined by \begin{align*}
    SCV_{n}(d,c):&=\frac{Var_{\theta^{(1)}}\left(\left.\frac{\mathbbm{1}_{A_n}(\theta^{(1)})/V(A_n)}{Z_n p(\theta^{(1)}|\mathcal{D})}\right|\mathcal{D}\right)}{E_{\theta^{(1)}}\left(\left.\frac{\mathbbm{1}_{A_n}(\theta^{(1)})/V(A_n)}{Z_n p(\theta^{(1)}|\mathcal{D})}\right|\mathcal{D}\right)^2}\\&=Z_n^2\left(\int_{A_n}\left(\frac{1/V(A_n)}{Z_np_n(\theta|\mathcal{D}_n)}\right)^2p_n(\theta|\mathcal{D}_n)\;d\theta-\frac{1}{Z_n^2}\right)=\int_{A_n}\frac{1/V(A_n)^2}{p_n(\theta|\mathcal{D}_n)}\;d\theta-1.
\end{align*}

Rescaling and shifting $A_n$ to the sphere \begin{align*}
    S=\{\theta|\theta^T\theta\leq c^2\}
\end{align*} gives \begin{align*}
    SCV_n(d,c)&=\int_{A_n}\frac{1/V(A_n)^2}{p_n(\theta|\mathcal{D}_n)}\;d\theta-1=\int_{A_n}\frac{1/\left(|\Sigma_n|V(S)^2\right)}{p_n(\theta|\mathcal{D}_n)}\;d\theta-1\\&=\int_{S}\frac{|\Sigma_n|^{\frac{1}{2}}/\left(|\Sigma_n|V(S)^2\right)}{p_n\left(\left.\Sigma_n^{\frac{1}{2}}\cdot\theta+m_n\right|\mathcal{D}_n\right)}\;d\theta-1=\frac{1}{V(S)^2}\int_{S}\frac{1}{|\Sigma_n|^{\frac{1}{2}}p_n\left(\left.\Sigma_n^{\frac{1}{2}}\cdot\theta+m_n\right|\mathcal{D}_n\right)}\;d\theta-1.
\end{align*}

Let $\|\cdot\|_{\infty,S}$ denote the supremum norm on the set $S$. By uniform convergence on the compact set $S$ \begin{align*}&|SCV_n(d,c)-SCV(d,c)|\\
    &=\frac{1}{V(S)^2}\left|\int_{S}\frac{1}{|\Sigma_n|^{\frac{1}{2}}p_n\left(\left.\Sigma_n^{\frac{1}{2}}\cdot\theta+m_n\right|\mathcal{D}_n\right)}\;d\theta-\int_{S}\frac{1}{|\Sigma|^{\frac{1}{2}}p\left(\left.\Sigma^{\frac{1}{2}}\cdot\theta+m\right|\mathcal{D}\right)}\;d\theta\right|\\&
    \leq\frac{1}{V(S)^2}\int_{S}\left|\frac{1}{|\Sigma_n|^{\frac{1}{2}}p_n\left(\left.\Sigma_n^{\frac{1}{2}}\cdot\theta+m_n\right|\mathcal{D}_n\right)}-\frac{1}{|\Sigma|^{\frac{1}{2}}p\left(\left.\Sigma^{\frac{1}{2}}\cdot\theta+m\right|\mathcal{D}\right)}\right|\;d\theta\\&\leq \frac{1}{V(S)}\left\|\frac{1}{|\Sigma_n|^{\frac{1}{2}}p_n\left(\left.\Sigma_n^{\frac{1}{2}}\cdot\theta+m_n\right|\mathcal{D}_n\right)}-\frac{1}{|\Sigma|^{\frac{1}{2}}p\left(\left.\Sigma^{\frac{1}{2}}\cdot\theta+m\right|\mathcal{D}\right)}\right\|_{\infty,S}\\&=\frac{1}{V(S)}\left\|\frac{|\Sigma_n|^{\frac{1}{2}}p_n\left(\left.\Sigma_n^{\frac{1}{2}}\cdot\theta+m_n\right|\mathcal{D}_n\right)-|\Sigma|^{\frac{1}{2}}p\left(\left.\Sigma^{\frac{1}{2}}\cdot\theta+m\right|\mathcal{D}\right)}{|\Sigma|^{\frac{1}{2}}p_n\left(\left.\Sigma^{\frac{1}{2}}\cdot\theta+m\right|\mathcal{D}_n\right)\cdot|\Sigma_n|^{\frac{1}{2}}p\left(\left.\Sigma_n^{\frac{1}{2}}\cdot\theta+m_n\right|\mathcal{D}\right)}\right\|_{\infty,S}\\&\leq \frac{1}{V(S)}\left\|\frac{1}{|\Sigma|^{\frac{1}{2}}p_n\left(\left.\Sigma^{\frac{1}{2}}\cdot\theta+m\right|\mathcal{D}_n\right)}\right\|_{\infty,S}\left\|\frac{1}{|\Sigma_n|^{\frac{1}{2}}p\left(\left.\Sigma_n^{\frac{1}{2}}\cdot\theta+m_n\right|\mathcal{D}\right)}\right\|_{\infty,S}\\&\cdot\left\||\Sigma_n|^{\frac{1}{2}}p_n\left(\left.\Sigma_n^{\frac{1}{2}}\cdot\theta+m_n\right|\mathcal{D}_n\right)-|\Sigma|^{\frac{1}{2}}p\left(\left.\Sigma^{\frac{1}{2}}\cdot\theta+m\right|\mathcal{D}\right)\right\|_{\infty,S}\stackrel{n\to\infty}\to0,
\end{align*}

since uniform convergence means convergence in the supremum norm and since the two reciprocals are uniformly bounded away from 0 on $S$ by the same reason. 

Let $C\subset(0,\infty)$ be compact. $1/V(S)$ is maximized by $c=\min C$ and the supremum is maximized by $c=\max C$. Denoting the resulting spheres by $S_{min},S_{max}$ respectively, gives \begin{align*}
    &\|SCV_n(d,c)-SCV(d,c)\|_{K,\infty} \\&\leq \frac{1}{V(S_{min})}\left\|\frac{1}{|\Sigma|^{\frac{1}{2}}p_n\left(\left.\Sigma^{\frac{1}{2}}\cdot\theta+m\right|\mathcal{D}_n\right)}\right\|_{\infty,S_{max}}\left\|\frac{1}{|\Sigma_n|^{\frac{1}{2}}p\left(\left.\Sigma_n^{\frac{1}{2}}\cdot\theta+m_n\right|\mathcal{D}\right)}\right\|_{\infty,S_{max}}\\&\cdot\left\||\Sigma_n|^{\frac{1}{2}}p_n\left(\left.\Sigma_n^{\frac{1}{2}}\cdot\theta+m_n\right|\mathcal{D}_n\right)-|\Sigma|^{\frac{1}{2}}p\left(\left.\Sigma^{\frac{1}{2}}\cdot\theta+m\right|\mathcal{D}\right)\right\|_{\infty,S_{max}}\stackrel{n\to\infty}\to0
\end{align*} for all $c\in C$. Thus the convergence is uniform on all compact subsets of $(0,\infty)$.

Let us fix $d$ $SCV_n(d,c)$ and restrict $c$ to an interval $[a,b], 0<a\leq c_d\leq b$. Then $SCV_n(d,c)$ converges uniformly in $c$ on $[a,b]$. Further, $SCV_n(d,c)$ is continuous in $c$ since an integral over an absolutely continuous function is continuous. Uniform convergence of continuous functions implies epigraphical convergence \cite[Proposition 7.15]{Rockafellar&2009}.  

$c_d$ is the unique minimal point of the limiting function $SCV(d,c)$ and the level sets $\{c\in[a,b]|SCV_n(d,c)\leq \alpha\}$ are bounded by definition. The statement in Theorem \ref{thm-n_limit_approx} follows by \cite[Theorem 7.33]{Rockafellar&2009}.

\end{proof}

\section*{Appendix 2: Derivations of Analytical Expressions for the Examples}

\noindent After the following proposition, all proofs for the multivariate Gaussian model of Section \ref{ssec:multivariateGaussian} are provided. 
\begin{proposition}\label{appendix:multivariateGaussian}

If $Y_i \in \mathbb{R}^{d}, i=1,\dots,n$, are drawn independently
from the multivariate normal distribution:
\begin{eqnarray*}
Y_i|\mu & \stackrel{\rm iid}{\sim} & {\rm MVN}_d(\mu, I_d), \;\; i=1,\ldots, n,
\end{eqnarray*}
and if the following prior distribution is considered for the mean vector $\mu$:
\begin{equation*}
    p(\mu)= {\rm MVN}_d(\mu; 0_d, s_0I_d),
\end{equation*}
with $s_0 > 0$, then the posterior distribution of the mean vector $\mu$ given the data $\mathcal{D}=\{y_1, \dots, y_n\}$ is given by:
\begin{equation*}
    p(\mu|\mathcal{D}) = {\rm MVN}_d(\mu; m_n,s_n I_d),
\end{equation*}
where $m_n=n\bar{y}/(n+1/s_0)$, $\bar{y}=(1/n)\sum_{i=1}^n y_i$, and $s_n=1/(n +1/s_0)$. Moreover, the marginal likelihood of the model can be written analytically as:
\begin{equation*}
    p(\mathcal{D}) = \prod_{j=1}^d {\rm MVN}_n (y_{.j}; 0_n, s_0 1_n 1_n^{\intercal} +   I_n),
\end{equation*}
where $y_{.j} \in \mathbb{R}^n$ is the vector of all observations for variable $j$ such that $[y_{.j}]_i = y_{ij}$ and $1_n$ is the vector of 1 in $\mathbb{R}^n$. 

\begin{proof}
    Let us first derive the analytical expression for the posterior distribution of the mean vector $\mu$:
    \begin{equation*}
        p(\mu | \mathcal{D}) \propto p(\mu)p(\mathcal{D}|\mu).
    \end{equation*}
    Taking the log of the expression and focusing on the terms in $\mu$:
    \begin{align*}
        \log p(\mu | \mathcal{D}) &= \log p(\mu) + \sum_{i=1}^{n}\log p(y_i|\mu) + \mathrm{cst}_1 \\
        &= -\frac{1}{2s_0}\mu^{\intercal}\mu  - \frac{1}{2}\sum_{i=1}^{n} (y_i-\mu)^{\intercal}(y_i - \mu) + \mathrm{cst}_2 \\
        &= -\frac{1}{2s_0}\mu^{\intercal}\mu  - \frac{1}{2}\sum_{i=1}^{n} \left(y_i^{\intercal}y_i - 2\mu^{\intercal}y_i + \mu^{\intercal}\mu  \right) + \mathrm{cst}_2 \\ 
        &= -\frac{1}{2}\mu^{\intercal}(\frac{I_d}{s_0} + nI_d)\mu + \mu^{\intercal}\sum_{i=1}^{n}y_i + \mathrm{cst}_3 \\
        &= -\frac{1}{2}\mu^{\intercal}\frac{1+n s_0}{s_0}I_d \mu + \mu^{\intercal} n \bar{y} + \mathrm{cst}_3.
        \end{align*}
    By identification, we recognize the functional form of a Gaussian distribution:
    \begin{equation*}
        p(\mu | \mathcal{D}) = {\rm MVN}_d(\mu; m_n, s_n I_d), 
    \end{equation*}
    where 
    \begin{align*}
        s_n &= \frac{s_0}{1 + n s_0} \\
        &= \frac{1}{n + 1/s_0},
        \end{align*}
    and
        \begin{align*}
            m_n &= s_n n \bar{y} \\
            &= \frac{1}{n + 1/s_0}n \bar{y}.
        \end{align*}

    To compute the marginal likelihood of the model, we first introduce the notation $Y_{.j}\in \mathbb{R}^{n}$ which corresponds to the vector of all observations for variable $j$ such that $[Y_{.j}]_i = Y_{ij}$. By construction:
    \begin{equation*}
        Y_{.j}| \mu_j \sim {\rm MVN}_n(\mu_j 1_n, I_n), \forall j \in \{1, \dots, d\},
    \end{equation*}
    where $\mu_j \in \mathbb{R}$ is the component $j$ of the vector $\mu$ and $1_n$ is the vector of 1 in $\mathbb{R}^n$. Note that the covariance matrix $I_n$ comes for the fact that all observations are independent with unit variance, given $\mu$. Interestingly, while the observations $(Y_i)_i$ are independent given the vector $\mu$, they are not marginally, and the marginal likelihood does not take a product form over marginal terms in $i$. Conversely, thanks to the isotropic Gaussian prior distribution which is considered for $\mu$, where the $(\mu_j)_j$ are all iid, not only are the vectors $(Y_{.j})_j$ independent given $\mu$, they are also independent marginally:
    \begin{align*}
        p(\mathcal{D}) &= \int p(y_{.1}, \dots, y_{.d}|\mu)p(\mu)d\mu \\
        &= \int \prod_{j=1}^{d}\left( p(y_{.j}|\mu_j)p(\mu_j) \right) d\mu \\
        &= \prod_{j=1}^{d} \int p(y_{.j}|\mu_j)p(\mu_j) d\mu_j \\
        &= \prod_{j=1}^{d} p(y_{.j}).
        \end{align*}
    Finally, to compute, $p(y_{.j})$, $Y_{.j}$ can be written as:
    \begin{equation*}
        Y_{.j} = \mu_j 1_n + \epsilon_j,
    \end{equation*}
    where 
    \begin{equation*}
      \epsilon_j \sim {\rm MVN}_n(0_n, I_n).  
    \end{equation*}
Thus, by construction, $Y_{.j}$ is defined as a product and sum over $\mu_j \sim \mathcal{N}(0, s_0)$ and $\epsilon_j \sim {\rm MVN}_n(0_n, I_n)$, which are independent from one another. Therefore, from Gaussian property:
   \begin{equation*}
        Y_{.j} \sim {\rm MVN}_n(0_n, s_0 1_n 1_n^{\intercal} + I_n).
    \end{equation*}
    Finally
    \begin{equation*}
        p(\mathcal{D})=\prod_{j=1}^{d} {\rm MVN}_n(y_{.j}; 0_n, s_0 1_n 1_n^{\intercal} + I_n).
    \end{equation*}
\end{proof}

\end{proposition}

\noindent After the following proposition, all proofs for the Bayesian linear regression model of Section \ref{ssec:bayesianRegression} are provided. 
\begin{proposition}\label{prop:linreg}
If a  linear regression model of the form
\begin{equation*}
    Y_i | x_i, \beta \sim \mathcal{N}(x_i^{\intercal} \beta, \sigma^2), i=1,\dots, n,
\end{equation*}
is considered where $Y_i \in \mathbb{R}, x_i \in \mathbb{R}^d, \beta \in \mathbb{R}^d, \sigma^2 \in \mathbb{R}$, with the variance $\sigma^2 > 0$ known, and if the following prior distribution is considered for the regression vector $\beta$:
\begin{equation*}
p(\beta) = {\rm MVN}_d(\beta; 0_d, I_d / \alpha),
\end{equation*}
with $\alpha > 0$, then the posterior distribution of the regression vector $\beta$ given the training data set $\mathcal{D} = \{(x_1, y_1), \dots, (x_n, y_n)\}$ is given by:
\begin{equation*}
    p(\beta| \mathcal{D}) = {\rm MVN}_d(\beta; m_n, \Sigma_n),
\end{equation*}
with
\begin{equation*}
    \Sigma_n^{-1} = \frac{X^{T}X}{\sigma^2} + \alpha I_d,
\end{equation*}
and
\begin{equation}
    m_n = (\alpha \sigma^2 I_d + X^{\intercal}X)^{-1} X^{\intercal} \textbf{y},
\end{equation}
where $\textbf{y} \in \mathbbm{R}^n$ is the vector of observed target variables $y_i$, and $X \in \mathcal{M}_{n \times d}(\mathbbm{R})$ is the design matrix where the input vectors $x_i \in \mathbbm{R}^d$ are stacked as row vectors. Moreover, the marginal likelihood of the model can also be written analytically as: 
\begin{equation*}
    p(\textbf{y}|X) = {\rm MVN}_n(\textbf{y}; O_n, \frac{XX^{T}}{\alpha} + \sigma^2 I_n).
\end{equation*}    

\begin{proof}

Relying on matrix notations, the linear regression model can be written as:
\begin{equation}
Y| X, \beta \sim {\rm MVN}_n(X\beta, \sigma^2 I_n),
\end{equation}
where $Y \in \mathbbm{R}^n$ is the random vector of target variables $Y_i$, and $X \in \mathcal{M}_{n \times d}(\mathbbm{R})$ is the design matrix where the input vectors $x_i \in \mathbbm{R}^d$ are stacked as row vectors. In such a supervised context, the training data set is made of all pairs $(x_i, y_i)$ and  can be denoted $\mathcal{D} = \{(x_1, y_1), \dots, (x_n, y_n)\}$. The posterior distribution of the regression vector $\beta$ is then given by:
\begin{align*}
    p(\beta | \mathcal{D}) &= p(\beta | (x_1, y_1), \dots, (x_n, y_n)) \\
    & \propto p(\beta) p\left((y_1, \dots, y_n) | (x_1, \dots, x_n), \beta \right). 
    \end{align*}
Now, denoting $\textbf{y} \in \mathbb{R}^n$ the \emph{observed} vector of target variables associated to $Y$,  taking the log of the expression, and focusing on the terms in $\beta$:
\begin{align*}
    \log p(\beta | \mathcal{D}) &= \log p(\beta) + \log p(\textbf{y} |X, \beta) + \mathrm{cst}_1 \\
        &= -\frac{\alpha}{2}\beta^{\intercal} \beta - \frac{1}{2\sigma^2} (\textbf{y} - X\beta)^{\intercal} (\textbf{y} - X\beta) + \mathrm{cst}_2 \\
        &= - \frac{\alpha}{2}\beta^{\intercal}\beta - \frac{1}{2\sigma^2}\left(\textbf{y}^{\intercal}\textbf{y} - 2 \beta^{\intercal}X^{\intercal} \textbf{y} + \beta^{\intercal}X^{\intercal} X \beta \right) + \mathrm{cst}_2 \\
        &= -\frac{1}{2}\beta^{\intercal} (\alpha I_d + \frac{X^{\intercal}X}{\sigma^2})\beta + \frac{\beta^{\intercal} X^{\intercal} \textbf{y}}{\sigma^2} + \mathrm{cst}_3.
        \end{align*}
By identification, we recognize the functional form of a Gaussian distribution: 
\begin{equation*}
    p(\beta | \mathcal{D}) = {\rm MVN}_d(\beta; m_n, \Sigma_n),
\end{equation*}
where 
\begin{equation*}
    \Sigma_n^{-1} = \frac{X^{\intercal}X}{\sigma^2} +  \alpha I_d,
\end{equation*}
and 
    \begin{align*}
        m_n &= \Sigma_n \frac{X^{\intercal} \textbf{y}}{\sigma^2} \\ 
            &= (\frac{X^{\intercal}X}{\sigma^2} + \alpha I_d)^{-1} \frac{X^{\intercal} \textbf{y}}{\sigma^2} \\
                &= (X^{\intercal}X + \alpha \sigma^2 I_d)^{-1} X^{\intercal} \textbf{y}.
    \end{align*}

Then, the marginal likelihood of the model can easily be obtained simply by writing:
\begin{equation*}
Y = X\beta + \epsilon,     
\end{equation*}
where 
\begin{equation*}
    \epsilon \sim {\rm MVN}_n(0_n, \sigma^2 I_n),
\end{equation*}
and 
\begin{equation*}
    \beta \sim {\rm MVN}_d(0_d, \frac{I_d}{\alpha}).
\end{equation*}
Thus, by construction, $Y$ is defined as a product and sum over the Gaussian random vectors $\beta$ and $\epsilon$, which are independent. Therefore, from Gaussian property:
\begin{equation*}
Y| X\sim {\rm MVN}_n(0_n, \frac{X X^{\intercal}}{\alpha} + \sigma^2 I_n),
\end{equation*}
and so the marginal likelihood is given by:
\begin{equation*}
p(\textbf{y}|X) = {\rm MVN}_n(\textbf{y}; 0_n, \frac{X X^{\intercal}}{\alpha} + \sigma^2 I_n).
\end{equation*}
\end{proof}

\end{proposition}

\end{document}